\documentclass[11pt,a4paper]{article}
\usepackage[utf8]{inputenc}
\usepackage[T1]{fontenc}
\usepackage[english]{babel}
\usepackage{tikz}\usetikzlibrary{calc,arrows}
\usepackage{geometry}
\usepackage{authblk}

\usepackage{amsmath}
\usepackage{amssymb}
\usepackage{amsthm}
\newtheorem{theorem}{Theorem}[section]

\newtheorem{definition}[theorem]{Definition}
\newtheorem{proposition}[theorem]{Proposition}

\newtheorem{example}[theorem]{Example}
\newtheorem{remark}[theorem]{Remark}
\newtheorem*{notation*}{Notation}
\makeatletter
\@addtoreset{equation}{section}
\makeatother

\usepackage{stmaryrd}
\usepackage{bbm}

\usepackage{graphicx}
\usepackage{microtype}
\usepackage[dvipsnames]{xcolor}
\usepackage[colorlinks=true,
pdfstartview=FitV,
linkcolor= MidnightBlue,
citecolor= MidnightBlue,
urlcolor= MidnightBlue,
hyperindex=true,
hyperfigures=false]
{hyperref}

\usepackage{cite}
\usepackage{appendix}

\newcommand{\rc}{\check{r}}
\newcommand{\Rc}{\check{R}}

\title{\textbf{
Twisted symmetric exclusion processes \\
and set-theoretical $R$-matrices}}
\author[1]{Mathieu Dabrowski\thanks{email: \href{mailto:mathieudabrowski@outlook.com}{mathieudabrowski@outlook.com}}\thanks{Author to whom any correspondence should be addressed.}}
\author[2]{Loïc Poulain d’Andecy\thanks{email: \href{mailto:loic.poulain-dandecy@univ-reims.fr}{loic.poulain-dandecy@univ-reims.fr}}}
\author[1]{Eric Ragoucy\thanks{email: \href{mailto:ragoucy@lapth.cnrs.fr}{ragoucy@lapth.cnrs.fr}}}
\affil[1]{Laboratoire d'Annecy de Physique Théorique, 9 Chemin de Bellevue - BP 110 - Annecy-le-Vieux - F-74941 Annecy Cedex - France}
\affil[2]{Laboratoire de Mathématiques de Reims, UMR CNRS 9008, Universit\'e Reims Champagne-Ardennes, Moulin de la Housse - BP 1039, 51687 Reims Cedex 2 - France}

\date{}
\begin{document}

\maketitle

\begin{abstract}

We investigate periodic integrable Markov models, constructed from set-theoretical solutions of the Yang-Baxter equation. We first focus on the simplest class of solutions, called Lyubashenko solutions. We show that the resulting models are equivalent to some twisted Symmetric Simple Exclusion Process (SSEP), which are usual multi-species periodic SSEP models where a twist is added on a bond of the ring. We also provide various possible interpretations for these Markov models. 

Then, we study the long time dynamics of the twisted SSEP, characterising its different stationary states and counting them. Allowing the twist to vary, we examine the possible transitions between the different stationary states. Finally, we extend our construction of Markov models to set-theoretical solutions that are more general than Lyubashenko solutions and show that such models are not equivalent to a twisted SSEP in general.

\end{abstract}


\section{Introduction}

Exclusion Markov processes were first introduced in the context of biology to model protein synthesis in \cite{mac1968} and in statistical physics to describe the dynamics of the Ising model in \cite{Kawa1966}. The mathematical foundations of these systems were established in \cite{Spitzer1970}.

Among these processes, the simplest one is  the Symmetric Simple Exclusion Process (SSEP). It describes a one-dimensional lattice gas, which is equivalent to a classical Ising model \cite{lee1952}, where  particles  have equal probabilities of jumping one site to the left or right, with the condition that the target site is empty (hard core interaction). On a periodic lattice, the SSEP reaches thermodynamic equilibrium, contrary to the Asymmetric Simple Exclusion Process (ASEP) where a macroscopic current appears \cite{Kirone2015}. Generalisations of those models with more than one  species of particles have been studied \cite{Evans1995,Karimipour1999,vanicat2017}.

The ASEP (respectively the SSEP) can be mapped to the anisotropic (respectively isotropic) Heisenberg quantum chain \cite{Schutz2001,Alcaraz1994}. More precisely, the Markov matrix containing the transition probabilities for the exclusion process is equivalent to the Hamiltonian of a quantum spin chain. This allows in some cases to use methods originally developed for quantum spin chains to study exclusion processes \cite{Golinelli2004,Gier2005}. 

Models constructed in this article have the property of being integrable. In the framework of quantum spin chains, the chain is said to be integrable if its Hamiltonian belongs to a set of commuting operators generated by a transfer matrix. A transfer matrix is constructed from an $ R $ matrix which is a solution of the Yang--Baxter equation (YBE). This integrability property allows the exact computation of physical quantities of the spin chain \cite{faddeev1996}. The definition of an integrable Markov process is obtained by replacing the Hamiltonian with the Markov matrix. Integrability can be used to obtain for instance information on the stationary state and on the dynamical behaviour of the stochastic process.
\medskip

In order to construct integrable Markov processes, the first step consists in finding solutions to the quantum Yang--Baxter equation \cite{Yang1967,Baxter1982Exactly}, which is a notoriously difficult problem in mathematical physics. To address it, Drinfeld in \cite{Drinfeld1992} suggested to consider ``set-theoretical'' solutions of the YBE, that is, solutions of the equation given by a permutation of a set of the form $ X \times X $. These solutions are not produced with the usual quantum group approach \cite{Isaev2004}. Since Drinfeld's proposal, set-theoretical solutions have been extensively studied \cite{etingof1998,Gateva1998,Gateva2004,Chouraqui2009,cedo2008,Rump2005,Rump2007,Cedo2012,Guarnieri2015,Doikou_2021,Doikou_2021_2,Smoktunowicz2018,Bachiller2016}. One simple class of set-theoretical solutions of the YBE is called Lyubashenko solutions \cite{Drinfeld1992}.

\medskip

In this paper, we construct and study new integrable Markov processes defined on a finite one dimensional lattice with periodic boundary conditions. These processes are constructed from set-theoretical solutions of the YBE using the Baxterization method \cite{Jones1990}. In particular, we consider involutive Lyubashenko solutions that allow the construction of models with several interesting interpretations.

We prove that those Lyubashenko solutions define Markov processes that are equivalent to some twisted periodic SSEP. The twisted SSEP is a generalisation of the multi-species SSEP on a ring where the exchange rule on one bond of the lattice allows particles to change species. In the context of quantum spin chains, systems with generalised periodic boundary conditions have been studied, such as in \cite{Batchelor1995,cao2013,Hao2016}, but they are less common in the framework of exclusion Markov processes.

Then we proceed to the study of the evolution of the twisted periodic SSEP for asymptotically long times. The process can be divided into sectors, for which there exists a unique stationary probability distribution. We define the profile and the total charge of a configuration and we prove that they uniquely label a sector. We also give explicit formulas for their number and size. Those formulas illustrate the idea that the twist connects different sectors of the multi-species SSEP.

Further we analyse how a stationary state of the twisted periodic SSEP is affected by a modification of the twist, in the same spirit as a quench procedure in condensed matter physics. Several interesting phenomena can occur such as the spreading and splitting of sectors. We define the branching probability of a stationary state going from one sector to another. Turning on and off the twist several times allows to alternate between different possible stationary states. We study in details some examples.

Finally, we consider more general set-theoretical solutions of the YBE than Lyubashenko solutions. We give an example of such solution and show that the Markov model constructed from it is not equivalent to any twisted SSEP coming from Lyubashenko solutions, thereby leaving room to more general models coming from set-theoretical solutions of the YBE which could be studied in further work.
\medskip

The plan of the paper goes as follows. Section \ref{sec_Msettheo} introduces periodic Markov models constructed from Lyubashenko solutions of the YBE that have the property of being integrable and that have various interpretations.

Section \ref{sec_MtwistedSSEP} presents the twisted periodic SSEP. In Section \ref{ssec:corr_set}, we prove the equivalence between models constructed from Lyubashenko solutions and some twisted periodic SSEP.

Section \ref{sec:comm_class_statio} focuses on sectors and stationary states of the twisted periodic SSEP. In Section \ref{ssec:statio}, we show that the probability distribution of the stationary states is uniform. We classify sector in Section \ref{subsec_invcomm} and in Sections \ref{ssec:comm_nb} and \ref{ssec:comm_card}, we count the number and find the sizes of sectors.

Section \ref{sec:twist_branch} analyses how stationary states of the twisted SSEP are affected by a modification of the twist. In Section \ref{ssec:spread_split_osci}, we describe possible behaviours that can occur when turning on and off the twist.

Section \ref{sec:set_theo_gene} introduces Markov models constructed from more general set-theoretical solutions than Lyubashenko solutions that are not equivalent to some twisted periodic SSEP in general. In Section \ref{ssec:non_eq}, we give an example of such model.

\section{Markov models from set-theoretical solutions}\label{sec_Msettheo}

In this section, we introduce Markov models associated to some set-theoretical solutions of the Yang--Baxter equation called Lyubashenko solutions. We show that those models are integrable and provide a non-exhaustive list of interesting interpretations.

\subsection{Markov models on a periodic lattice}\label{ssec:mark_perio}

We consider a one-dimensional periodic lattice composed of $ L $ sites indexed by the variable $ i\in \llbracket 1,L \rrbracket $. On each site $ i $, the local configuration variable $\tau_i$ takes values in $\llbracket 0,N-1 \rrbracket$ with $N$ a positive integer. 
The configurations $ \mathfrak{C} $ of the system correspond to $ L $-tuples $\vec \tau= (\tau_{1},...,\tau_{L})\in \llbracket 0,N-1 \rrbracket^{L} $. The canonical basis vector of a configuration $ \vec \tau $ is denoted $ \lvert\vec \tau \rangle $. All together they form an orthonormal basis i.e.  $\langle\vec\tau \mid\vec\tau' \rangle=\delta_{\vec\tau,\vec\tau'}$.

The stochastic dynamics of the model is given by the Markov matrix
\begin{equation}M=\sum_{\vec \tau,\vec \tau'\in \mathfrak{C}}m(\vec \tau' \rightarrow \vec \tau)\,\lvert \vec \tau\rangle \langle \vec \tau'\rvert\, ,\end{equation}
where $ m(\vec \tau' \rightarrow \vec \tau)\geq 0 $ for $ \vec \tau\neq \vec \tau' $ is interpreted as a transition rate between configurations $ \vec \tau' $ and $ \vec \tau $, and by definition, we have 
\begin{equation}\label{coeffdiagM}m(\vec \tau \rightarrow \vec \tau)=-\sum_{\vec \tau'\in\mathfrak{C},\vec \tau'\neq \vec \tau}m(\vec \tau \rightarrow \vec \tau')\ .\end{equation} 
 Starting from an initial probability $P_0$ on the configuration space, the probability of the system to be in configuration $ \vec \tau $ at continuous time $ t $ is denoted $ P_{t}(\vec \tau) $. The probabilities of all configurations at time $ t $ are stored in the vector $ \lvert P_{t}\rangle = \sum_{\vec \tau } P_{t}(\vec \tau) \lvert \vec \tau\rangle$.
The evolution of $\lvert P_{t}\rangle$ is governed by the master equation 
\begin{equation}
\mathrm{d}\lvert P_{t}\rangle/\mathrm{dt}=M\lvert P_{t}\rangle .
\end{equation}

We will consider Markov matrices which are local, by which we mean of the following form:
\begin{equation}\label{eq:M_model}
	M=\sum_{i=1}^{L-1}m_{i,i+1}+\tilde{m}_{L,1}\, ,
\end{equation}
where $m$ is a local jump operator acting on two sites, and $m_{i,i+1}$ is its action on sites $ i $ and $ i+1 $. The additional term $\tilde{m}_{L,1}$ is a local jump operator $\tilde{m}$ acting on sites $L$ and $1$. When $\tilde{m}=m$, the model is  periodic. If $\tilde{m}_{L,1}=\mathcal{T}_L\,m_{L,1}\,\mathcal{T}_L^{-1}$ for some $N$-dimensional matrix $\mathcal{T}$, the model is called twisted periodic (by the twist $\mathcal{T}$).

\subsection{Lyubashenko solutions of the Yang-Baxter equation}\label{ssec:lyub_sol}

Fix throughout this section a bijective map denoted:
\[g\ :\ \llbracket 0,N-1\rrbracket\to\llbracket 0,N-1\rrbracket\ .\] 
From $g$ we define the following operator on two sites:
\begin{equation}\label{eq:r_lyub1}
\check r=\sum_{i,j\in \llbracket 0,N-1\rrbracket}|g(j),g^{-1}(i)\rangle\langle i,j|\ .\end{equation}
It satisfies the braided Yang--Baxter equation:
\begin{equation}\label{eq:cYB}
\check r_{12}\,\check r_{23}\,\check r_{12}=\check r_{23}\,\check r_{12}\,\check r_{23}\ ,
\end{equation}
and is called a Lyubashenko solution of the Yang--Baxter equation. Further, it is easy to check that it is involutive in the sense that
\begin{equation}\label{eq:invo}
\check r^{2}=\mathrm{Id}\ .\end{equation}

\begin{remark}
The solution $\check r$ is set-theoretical because it comes from a bijective map on the set $\llbracket 0,N-1 \rrbracket^2$, namely the Yang--Baxter map is given by $(x,y)\mapsto (g(y),g^{-1}(x))$.
\end{remark}

From $\check r$, we define the local jump operator 
\begin{equation}\label{eq:loc_jump}
m=\check r-\mathrm{Id}\,,
\end{equation}
and the Markov matrix
\begin{equation}\label{eq:M_Lyub}
	M=\sum_{i=1}^{L-1}m_{i,i+1}+m_{L,1}\ .
\end{equation}
It is easy to see that $M$ is indeed a Markov matrix, namely that its off-diagonal coefficients are non-negative and that each column sums to $0$. Note that this is a periodic model. 

Also $\check r$ is symmetric and therefore $M$ too, which is expressed in terms of the transition rates as:
\[m(\vec \tau \rightarrow \vec \tau')=m(\vec \tau' \rightarrow \vec \tau)\,,\ \ \ \ \forall \vec \tau,\vec \tau'\in\mathfrak{C}\ .\]

\subsection{Integrability of the model}\label{ssec:int}

Here we show that the model is integrable in the sense that the Markov matrix belongs to a set of commuting operators generated by a transfer matrix.
We follow the approach developed by the St Petersburg school on quantum integrable systems, see e.g. \cite{faddeev1996,FST}.

The matrix $\check r$ can be Baxterized as follows:
\begin{equation}\label{eq:R}
	\check R(z)=\frac{z\, \rc+ \mathrm{Id}}{z+1}\ .
\end{equation}
Indeed, since $ \rc^{2}=\mathrm{Id} $, it is easy to check that the Yang--Baxter equation with spectral parameters is satisfied:
\begin{equation}\label{eq:YBE-z}\check R_{12}(u)\,\check R_{23}(u+v)\,\check R_{12}(v)=\check R_{23}(v)\,\check R_{12}(u+v)\,\check R_{23}(u)\ .\end{equation}
The transfer matrix is constructed with the formula
\begin{equation}
	t(z)=\mathrm{tr}_{0}\, R_{0,L}(z)\,R_{0,L-1}(z)\cdots R_{0,1}(z)\, ,
\end{equation}
where $R(z)=P\check R(z)$ with $P$ the permutation operator. The Yang--Baxter equation \eqref{eq:YBE-z} ensures that 
$[t(z)\,,\,t(z')]=0$, so that upon expansion in $z$, $t(z)$ indeed generates a set of commuting operators.
Moreover, since $R(0)=P$, the trace in $t(0)$ and $t'(0)$ can be easily computed, and we recover with a straightforward calculation the Markov matrix of (\ref{eq:M_Lyub}) as
\begin{equation}\label{eq:M_lyub2}M=t(0)^{-1}t'(0)=\sum_{i=1}^{L-1}m_{i,i+1}+m_{L,1}\ .
\end{equation}

\subsection{Interpretations of the models}\label{ssec:interpret}

In the next section, the Markov model above will be transformed to a twisted SSEP model, closely related to well-studied models. However, we suggest here possible direct interpretations of the Markov model constructed above.

Recall that we are given a bijection $g$ of the set $\llbracket 0,N-1 \rrbracket$, in terms of which the transition rates imply local processes on two sites of the form
\begin{equation}\label{localprocess}(\,\tau_{i}\,,\,\tau_{i+1}\,)\ \ \to\ \ (\,g(\tau_{i+1})\,,\,g^{-1}(\tau_{i}\,)\,)\ \ \ \ i=1,\dots,L,\end{equation}
where of course $L+1$ means $1$.
\begin{remark}
If $g$ is the trivial bijection (the identity map) then the model is simply the multi-species SSEP.
\end{remark}

\subsubsection{Symmetric exclusion process with oscillations}

The first natural interpretation is the usual one, where we have particles of $ N $ different species labeled by a variable $ s\in \llbracket 0,N-1 \rrbracket $ which evolve on this periodic lattice. The particles are subjected to a hard core interaction, i.e. each site is occupied by at most one particle.

Note however that starting from a state with two species $a,b$ located in adjacent sites, (\ref{localprocess}) says that it may jump to a state with two species $g(b),g^{-1}(a)$ located in these sites. Roughly speaking, the particles are oscillating between different species while moving on the lattice.

\subsubsection{Species with internal states}

Here we make use of the cyclic decomposition of the bijection $g$ to provide different interpretations. We denote the decomposition of the permutation $g$ into disjoint cycles by 
\[g=\pi_1\dots\pi_{n}\, ,\]
where $\pi_1,\dots,\pi_{n}$ are cycles with disjoint support in $\llbracket 0,N-1\rrbracket$. We call $c_i$ the length of the cycle $\pi_i$. Clearly, since $g$ is a permutation of $\llbracket 0,N-1 \rrbracket$ we must have $\sum_{i=1}^nc_i=N$. We reindex the set $\llbracket 0,N-1\rrbracket$ to make use of such a decomposition. We set:
\begin{equation*}
\llbracket 0,N-1 \rrbracket=\{1^{(0)},\dots,1^{(c_1-1)},2^{(0)},\dots,2^{(c_2-1)},\dots\dots,n^{(0)},\dots,n^{(c_n-1)}\}\, ,
\end{equation*}
so that the permutation becomes in cyclic notation:
\begin{equation*}
g=(1^{(0)},\dots,1^{(c_1-1)})(2^{(0)},\dots,2^{(c_2-1)})\dots\dots(n^{(0)},\dots,n^{(c_n-1)})\, .
\end{equation*}
Now the interpretation is that we have $n$ different species and for each species $s$, we have an internal state labelled by a ``quantum number'', that we will call for brevity a ``charge'', which takes values in $\{0,\dots,c_s-1\}$. With these notations, the action of $g$ does not transform species, it  simply increases the value of the charge by $1$ (while, obviously, $g^{-1}$ reduces the charge by $1$). The local processes are now of the form
\begin{equation}\label{localprocess2}{(\,s^{(e)}\,,\,t^{(e')}\,)\ \ \to\ \ (\,t^{(e'+1)}\,,\,s^{(e-1)}\,)\ .}\end{equation}
Note that for species $s$, the value of the charge is understood modulo $c_s$.

\begin{example}
In this  interpretation, a particularly simple case is when $g$ has only one cycle: $g=(1^{(0)},\dots,1^{(N-1)})$. Then we have only one species on our lattice which carries a quantum number, or a charge, taking values in the integers modulo $N$.
\end{example}

\begin{remark}
It may happen that there is one cycle (or more) in $g$ of length one. The corresponding species then carries no charge and will be called a  neutral species. Any neutral species can be selected as the vacuum. If there is more than one neutral species, up to a relabeling of the neutral species, any choice leads to the same physical model.
\end{remark}

\subsubsection{Rolling polygons}

Another possible interpretation is closely related to the previous one. One can see each species $s$ as a regular polygon with $c_s$ sides, each side being labeled by an integer which corresponds to the 'charge' introduced in the previous interpretation. In this context, the polygons are rolling clockwise when moving in one direction of the lattice and rolling counterclockwise in the opposite direction, as illustrated in Figure 1.

\begin{center}
\begin{tikzpicture}[x=1cm,y=1cm, font=\small, line join=round, line cap=round]
\tikzset{
  poly/.style={thick},
  support/.style={thick},
  arr/.style={-latex, thick},
  lab/.style={font=\scriptsize},
  lab2/.style={font=\large}
}

\def\s{1.3}                  
\def\t{1.3}                  
\def\h{0.25}                 
\def\gap{0.}               
\def\off{0.15}               

\begin{scope}[shift={(0,0)}]

  \draw[support] (-0.2,0) -- (4.2,0);
  \draw[support] (-0.2,0) -- (-0.2,\h);
  \draw[support] (2.0,0) -- (2.0,\h);
  \draw[support] (4.2,0) -- (4.2,\h);

  \coordinate (A1) at (1.0-\t/2,\h+\gap);
  \coordinate (A2) at (1.0+\t/2,\h+\gap);
  \coordinate (A3) at (1.0, \h+\gap + 0.866*\t);  

  \draw[poly] (A1)--(A2)--(A3)--cycle;

  \node[lab] at ($(A1)!0.5!(A3)+(-\off,0)$) {1};
  \node[lab] at ($(A2)!0.5!(A3)+(\off,0)$)  {2};
  \node[lab] at ($(A1)!0.5!(A2)+(0,\off)$)  {3};
  \draw[arr] (1.0,1.2) -- (1.7,1.2);

  \draw[arr] (0.3,1.3) arc[start angle=180, end angle=20, radius=0.8];

  \coordinate (B1) at (2.4,\h+\gap);
  \coordinate (B2) at ($(B1)+(\s,0)$);
  \coordinate (B3) at ($(B2)+(0,\s)$);
  \coordinate (B4) at ($(B1)+(0,\s)$);

  \draw[poly] (B1)--(B2)--(B3)--(B4)--cycle;

  \node[lab] at ($(B4)!0.5!(B3)+(0,\off)$)  {3};    
  \node[lab] at ($(B4)!0.5!(B1)+(-\off,0)$) {2};    
  \node[lab] at ($(B2)!0.5!(B3)+(\off,0)$)  {4};    
  \node[lab] at ($(B1)!0.5!(B2)+(0,\off)$)  {1};    

  \draw[arr] (4.1,1.5) arc[start angle=0, end angle=160, radius=0.8];

  \draw[arr] (2.5,1.4) -- (1.9,1.4);

\end{scope}

\draw[arr] (4.9,1.0) -- (6.4,1.0);

\begin{scope}[shift={(7.0,0)}]

  \draw[support] (-0.2,0) -- (4.2,0);
  \draw[support] (-0.2,0) -- (-0.2,\h);
  \draw[support] (2.0,0) -- (2.0,\h);
  \draw[support] (4.2,0) -- (4.2,\h);

  \coordinate (C1) at (0.4,\h+\gap);
  \coordinate (C2) at ($(C1)+(\s,0)$);
  \coordinate (C3) at ($(C2)+(0,\s)$);
  \coordinate (C4) at ($(C1)+(0,\s)$);

  \draw[poly] (C1)--(C2)--(C3)--(C4)--cycle;

  \node[lab] at ($(C4)!0.5!(C3)+(0,\off)$)  {4};
  \node[lab] at ($(C4)!0.5!(C1)+(-\off,0)$) {3};
  \node[lab] at ($(C2)!0.5!(C3)+(\off,0)$)  {1};
  \node[lab] at ($(C1)!0.5!(C2)+(0,\off)$)  {2};

  \coordinate (D1) at (2.4,\h+\gap);
  \coordinate (D2) at ($(D1)+(\t,0)$);
  \coordinate (D3) at ($(D1)+(0.5*\t,0.866*\t)$);

  \draw[poly] (D1)--(D2)--(D3)--cycle;

  \node[lab] at ($(D1)!0.5!(D3)+(-\off,0)$) {3};
  \node[lab] at ($(D2)!0.5!(D3)+(\off,0)$)  {1};
  \node[lab] at ($(D1)!0.5!(D2)+(0,\off)$)  {2};

\end{scope}
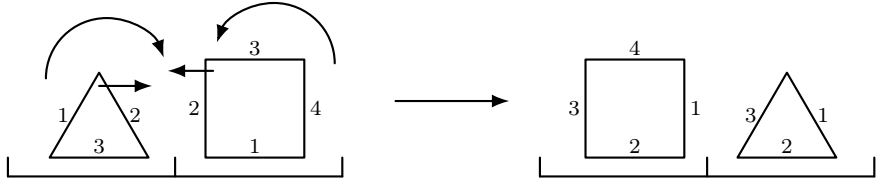
\node[lab2, align=center, anchor=north] at ($(current bounding box.south)+(0,-0.5)$) {\textsl{Figure 1: Transitions for rolling triangles and squares.}};
\end{tikzpicture}
\end{center}

For example, if $g$ consists of a single full cycle, $g=(0,1,\dots,N-1)$, then the lattice is filled by regular $N$-gons.

\subsubsection{Filling and emptying colored boxes}

Finally one can also interpret a local configuration of the form $s^{(e)}$ as a box ``of color $s$'' containing a number $e$ of items (particles, balls, etc.).

In this interpretation, we have boxes moving on the periodic lattice. The boxes are labeled or colored by an integer $1,\dots,n$ and a box of type $s$ may contain up to $c_s$ items. When two adjacent boxes are switched on the lattice, one acquires one more item and the other one loses one item.

Again, the number of items inside a box of color $s$ is considered modulo $c_s$. This means in particular that a box with $c_s-1$ items becomes empty when acquiring an additional item (which can be seen as a discharging of the box when full). More delicate to interpret is when an empty box (with $0$ item) ``loses'' an item and becomes a box containing $c_s-1$ items. This interpretation is illustrated in Figures 2 and 3.
\begin{center}
\begin{tikzpicture}[>=stealth, line cap=round, line join=round]
  \def\TallH{3.cm}
  \def\ShortH{2.25cm}
  \def\BoxW{0.9cm}

  \tikzset{
    tallbox/.style={draw, thick, minimum width=\BoxW, anchor=south, minimum height=\TallH},
    shortbox/.style={draw, thick, minimum width=\BoxW, anchor=south, minimum height=\ShortH},
    emptydot/.style={draw, thick, circle, minimum size=6mm, inner sep=0pt},
    fulldot/.style={draw, thick, fill=black!60, circle, minimum size=6mm, inner sep=0pt},
    lab/.style={font=\large}
  }

  \node[tallbox]  (L4) at (1.4,0) {};
  \node[emptydot] at ($(L4.south)+(0,0.4)$) {};
  \node[emptydot] at ($(L4.south)+(0,1.12)$) {};
  \node[lab]      at ($(L4.south)+(-0.6,2.7)$) {\small 4};

  \node[shortbox] (L3) at (0,0) {};
  \node[fulldot]  at ($(L3.south)+(0,0.4)$) {};
  \node[fulldot]  at ($(L3.south)+(0,1.12)$) {};
  \node[lab]      at ($(L3.south)+(-0.6,2.0)$) {\small 3};

  \node[tallbox]  (R4) at (4.15,0) {};
  \node[emptydot] at ($(R4.south)+(0,0.4)$) {};
  \node[emptydot] at ($(R4.south)+(0,1.12)$) {};
  \node[emptydot] at ($(R4.south)+(0,1.83)$) {};
  \node[lab]      at ($(R4.south)+(-0.6,2.7)$) {\small 4};

  \node[shortbox] (R3) at (5.55,0) {};
  \node[fulldot]  at ($(R3.south)+(0,0.4)$) {};
  \node[lab]      at ($(R3.south)+(-0.6,2.0)$) {\small 3};

  \draw[<->, thick] (2.2, 1.12) -- (3.35, 1.12);

  \node[align=center, anchor=north] at ($(current bounding box.south)+(0,-0.5)$) 
  {
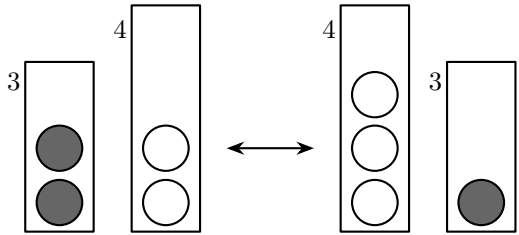
    \textsl{Figure 2: generic transitions in the colored boxes interpretation.} \\ 
    \textsl{The number on the left of a box indicates its maximal number of balls.}
  };

\end{tikzpicture}
\end{center}

\medskip
\begin{center}
\begin{tikzpicture}[>=stealth, line cap=round, line join=round]
  \def\TallH{3.cm}
  \def\ShortH{2.25cm}
  \def\BoxW{0.9cm}
  \def\shfthor{8} 

  \tikzset{
    tallbox/.style={draw, thick, minimum width=\BoxW, anchor=south, minimum height=\TallH},
    shortbox/.style={draw, thick, minimum width=\BoxW, anchor=south, minimum height=\ShortH},
    emptydot/.style={draw, thick, circle, minimum size=6mm, inner sep=0pt},
    fulldot/.style={draw, thick, fill=black!60, circle, minimum size=6mm, inner sep=0pt},
    lab/.style={font=\large}
  }


  \node[tallbox]  (L5) at (1.35,0) {};
  \node[emptydot] at ($(L5.south)+(0,0.4)$) {};
  \node[emptydot] at ($(L5.south)+(0,1.12)$) {};
  \node[emptydot] at ($(L5.south)+(0,1.83)$) {};
  \node[emptydot] at ($(L5.south)+(0,2.54)$) {};
  \node[lab]      at ($(L5.south)+(-0.6,2.7)$) {\small 4};

  \node[shortbox] (L6) at (0,0) {};
  \node[fulldot]  at ($(L6.south)+(0,0.4)$) {};
  \node[fulldot]  at ($(L6.south)+(0,1.12)$) {};
  \node[lab]      at ($(L6.south)+(-0.6,2.0)$) {\small 3};

  \node[shortbox] (R5) at (5.5,0) {};
  \node[fulldot]  at ($(R5.south)+(0,0.4)$) {};
  \node[lab]      at ($(R5.south)+(-0.6,2.0)$) {\small 3};

  \node[tallbox]  (R6) at (4.15,0) {};
  \node[emptydot] at ($(R6.south)+(0,0.4)$) {};
  \node[lab]      at ($(R6.south)+(-0.6,2.7)$) {\small 4};

  \draw[<->, thick] (2.15,1.12) -- (3.35,1.12);
  \node[lab] at (2.75,-0.3) {(3.a)};


  \node[tallbox]  (L4) at ($(1.35,0)+(\shfthor,0)$) {};
  \node[emptydot] at ($(L4.south)+(0,0.4)$) {};
  \node[emptydot] at ($(L4.south)+(0,1.12)$) {};
  \node[lab]      at ($(L4.south)+(-0.6,2.7)$) {\small 4};

  \node[shortbox] (L3) at ($(0,0)+(\shfthor,0)$) {};
  \node[fulldot]  at ($(L3.south)+(0,0.4)$) {};
  \node[lab]      at ($(L3.south)+(-0.6,2.0)$) {\small 3};

  \node[shortbox] (R3) at ($(5.5,0)+(\shfthor,0)$) {};
  \node[fulldot]  at ($(R3.south)+(0,0.4)$) {};
  \node[fulldot]  at ($(R3.south)+(0,1.12)$) {};
  \node[fulldot]  at ($(R3.south)+(0,1.83)$) {};
  \node[lab]      at ($(R3.south)+(-0.6,2.0)$) {\small 3};

  \node[tallbox]  (R4) at ($(4.15,0)+(\shfthor,0)$) {};
  \node[emptydot] at ($(R4.south)+(0,0.4)$) {};
  \node[emptydot] at ($(R4.south)+(0,1.12)$) {};
  \node[emptydot] at ($(R4.south)+(0,1.83)$) {};
  \node[lab]      at ($(R4.south)+(-0.6,2.7)$) {\small 4};

  \draw[<->, thick] ($(2.15,1.12)+(\shfthor,0)$) -- ($(3.35,1.12)+(\shfthor,0)$);
  \node[lab] at ($(2.75,-0.3)+(\shfthor,0)$) {(3.b)};

  \node[align=center, anchor=north] at ($(current bounding box.south)+(0,-0.4)$) 
  {
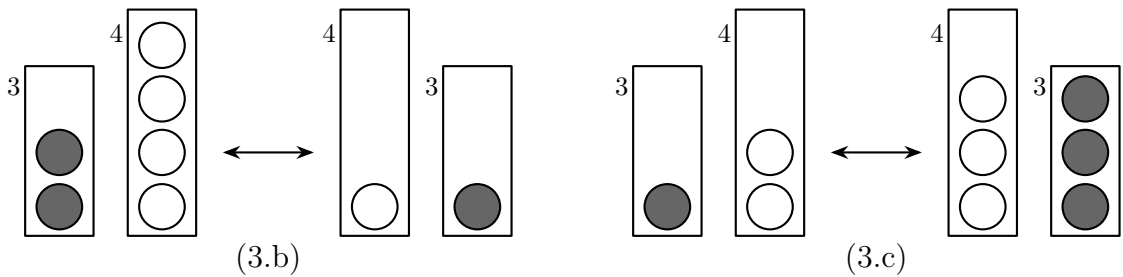
    \textsl{Figure 3: Overfilling (3.a) / emptying (3.b) boxes.} \\
    \textsl{Special transitions in the colored boxes interpretation.}
  };

\end{tikzpicture}
\end{center}

Note that if $g$ consists of a single full cycle, $g=(0,1,\dots,N-1)$, then we have only one type of boxes which contain up to $N-1$ items (seen modulo $N$).

\section{Twisted SSEP}\label{sec_MtwistedSSEP}

In this section, we show that the Markov models from the previous section are mathematically equivalent to some twisted periodic SSEP models. 
We first define and interpret these models, show directly their integrability, and then proceed to the description of the correspondence with the previous section.

\subsection{Definition of the model}\label{ssec:def_tw_SSEP}

As in the preceding section, our model will depend on a bijection denoted here \[f\ :\ \llbracket 0,N-1\rrbracket\to\llbracket 0,N-1\rrbracket\, .\] 
\begin{remark}
It will be explicitly related (but not equal) to the bijection $g$ of the preceding section.
\end{remark}
For the Markov matrix, we take the following:
\begin{equation}\label{eq:Mf_model}
	M_f=\sum_{i=1}^{L-1}m_{i,i+1}+m_{L,1}^{f}\ .
\end{equation}
 Here, for $ i\in \llbracket 1,L-1 \rrbracket $, the local jump operator $ m_{i,i+1}$ acting on sites $i$ and $i+1$ corresponds to the multi-species SSEP operator
\begin{equation}\label{eq:local_m}
	m_{i,i+1}=\sum_{\tau_{i},\tau_{i+1}\in \llbracket 0,N-1 \rrbracket}\lvert \tau_{i+1},\tau_{i} \rangle \langle \tau_{i},\tau_{i+1} \rvert-\mathrm{Id}\ .
\end{equation}
The local jump operator between sites $L$ and $1$ is given by
\begin{equation}\label{eq:local_mtw}	m_{L,1}^{f}=\sum_{\tau_{L},\tau_{1}\in \llbracket 0,N-1 \rrbracket}\lvert f(\tau_{1}),f^{-1}(\tau_{L}) \rangle \langle \tau_{L},\tau_{1} \rvert-\mathrm{Id}\ .
\end{equation}
Identifying the bijection $f$ with the linear operator acting on one site:
\[f=\sum_{\tau\in\llbracket 0,N-1\rrbracket}|f(\tau)\rangle\langle \tau|\,,\]
the local jump operator $m_{L,1}^{f}$ is of the form
\[m_{L,1}^{f}=f_L\,m_{L,1}\,f_{L}^{-1}\ ,\]
so that we have a twisted periodic model, where the local SSEP operator between sites $L$ and $1$ is twisted by the bijection $f$.

As in the preceding section, it is easy to see that $M_f$ is a symmetric matrix.

\subsection{Interpretations}\label{ssec:interpretations}

The twist $f$ added between sites $L$ and $1$ has the same various interpretations as those given in the preceding sections. Here note that the dynamics between sites $i$ and $i+1$ for $i=1,\dots,L-1$ is the usual one for SSEP, given by
\[(\,\tau\,,\,\tau'\,)\ \ \to\ \ (\,\tau'\,,\,\tau\,)\ \, ,\]
while between sites $L$ and $1$, it is given by
\[(\,\tau\,,\,\tau'\,)\ \ \to\ \ (\,f(\tau')\,,\,f^{-1}(\tau)\,)\ .\]
Therefore, if we think about $\llbracket 0,N-1\rrbracket$ as representing $N$ different species, the species are affected by the dynamics only between sites $L$ and $1$.

As before, reindexing the set $\llbracket 0,N-1\rrbracket$ to make use of the cyclic decomposition of $f$, we set
\[f=(1^{(0)},\dots,1^{(c_1-1)})(2^{(0)},\dots,2^{(c_2-1)})\dots\dots(n^{(0)},\dots,n^{(c_n-1)})\, ,\]
and we use the interpretation of $n$ different species carrying an internal degree of freedom, that we call a ``charge''.

Then again the dynamics between sites $i$ and $i+1$ for $i=1,\dots,L-1$ is the usual one for SSEP (simple flip of the two local entries) while between sites $L$ and $1$, the charges are changing according to:
\begin{equation}\label{localprocess2b}(\,s^{(e)}\,,\,t^{(e')}\,)\ \ \to\ \ (\,t^{(e'+1)}\,,\,s^{(e-1)}\,)\ .\end{equation}
Thus, a non-trivial function $f$ corresponds to a purely transmitting impurity (or an ``electric'' field) located between sites $L$ and 1 and which affects the ``charge'' of the particles when they pass through.

As before, the various interpretations with polygons rolling or with colored boxes filling and emptying are possible. The new feature compared to the preceding section is that the rolling of polygons, or the filling/emptying of boxes only happens when crossing the boundary between sites $L$ and $1$.

\subsection{Integrability}\label{ssec:int_twist_SSEP}

We proceed to show that the above models are also integrable by adding the twist in the construction of the transfer matrix. Again, this property is very well-known in the context of quantum integrable systems, see e.g. \cite{faddeev1996,FST}.

The Baxterized $R$-matrix this time is given by
\begin{equation}\label{eq:R_twist}
	R(z)=\frac{z\mathrm{Id}+P}{z+1}\, ,
\end{equation}
where $P$ is the flip operator, and satisfies the Yang Baxter equation with spectral parameter : $$ R_{12}(u)R_{13}(u+v)R_{23}(v)=R_{23}(v)R_{13}(u+v)R_{12}(u)\ .$$
In general, an integrable twist operator is an invertible matrix $ \mathcal{T}(z) $ satisfying
\begin{equation}\label{eq:int_twist}
	\Rc(u-v)\mathcal{T}(u) \otimes \mathcal{T}(v)= \mathcal{T}(v) \otimes  \mathcal{T}(u) \Rc(u-v)\ ,
\end{equation}
where  $\Rc(z)=PR(z)$. The homogeneous twisted transfer matrix is given by (see e.g. \cite{vanicat2017})
\begin{equation}
	t_{\mathcal{T}}(z)=\mathrm{tr}_{0} R_{0,L}(z)R_{0,L-1}(z)\cdots R_{0,1}(z)\mathcal{T}_{0}(z)\, ,
\end{equation}
and satisfies the commutativity relation $[t(y),t(z)]=0$. The general expression of the twisted Markov matrix is 
\begin{equation}\label{eq:M_twist}
	M_{\mathcal{T}}=t_{\mathcal{T}}(0)^{-1}t'_{\mathcal{T}}(0)=\sum_{i=1}^{L-1}m_{i,i+1}+\mathcal{T}_{L}^{-1}(0)m_{L,1}\mathcal{T}_{L}(0)+\mathcal{T}_{L}(0)^{-1}\mathcal{T}'_{L}(0)\, ,
\end{equation}
In our case, we simply have to take
$$\mathcal{T}(z)=f^{-1}=\sum_{\tau\in\llbracket 0,N-1\rrbracket}|f^{-1}(\tau)\rangle\langle \tau|\,,$$ where again we have identified the bijection $f^{-1}$ with the corresponding linear operator. We have that $\mathcal{T}(z)$ is indeed an integrable twist since it is invertible and equation \eqref{eq:int_twist} is satisfied because the twist is independent on $ z $ and because of the form of $ R(z) $ in equation \eqref{eq:R_twist}.
We have by identifying equation \eqref{eq:Mf_model} and equation \eqref{eq:M_twist} that
\begin{equation}\label{eq:m_L1}
	\mathcal{T}_{L}^{-1}(0)m_{L,1}\mathcal{T}_{L}(0)=f_Lm_{L,1}f_L^{-1}=m_{L,1}^{f}\, .
\end{equation}
Therefore, the expression (\ref{eq:M_twist}) produces in this case the Markov matrix \eqref{eq:Mf_model} of our model and this shows the integrability of the model.

\subsection{Correspondence with set-theoretical models}\label{ssec:corr_set}

For this subsection, we denote by $\tilde{M}_g$ the Markov matrix \eqref{eq:M_Lyub} from Section \ref{sec_Msettheo} constructed from the set-theoretical solution of the Yang--Baxter equation corresponding to the bijection $g$. 

Here we show that $\tilde{M}_g$ is conjugated to a Markov matrix $M_f$ in \eqref{eq:Mf_model} of the twisted SSEP, and moreover that this conjugation is simply done by a bijection on the configurations.

Recall that the set of configurations is
\[\mathfrak{C}=\{\vec{\tau}=(\tau_1,\dots,\tau_L)\ ,\ \ \tau_i\in\llbracket 0,N-1\rrbracket\}\, .\]
In the following proposition, we make again the identification of a bijection of $\mathfrak{C}$ with the linear operator acting on the vector space of the model. \begin{proposition}\label{prop:corr_set_twist}
Let $f=g^L$. We have:
\begin{equation}\label{eq_conjugationMarkov}
	\tilde{M}_g=V^{-1}M_fV\, ,
\end{equation}
where $V$ is the operator acting on the $L$ sites of the lattice as 
\begin{equation}\label{eq:defV}
    V=\mathrm{id}\otimes g \otimes g^2 \otimes\dots\otimes g^{L-1}\,.
\end{equation}
\end{proposition}

\begin{proof}
It is well known (see for example \cite{Doikou_2021} for a more general statement) and easy to check that the set-theoretical solution
\begin{equation}\check r=\sum_{i,j\in\llbracket 0,N-1\rrbracket}|g(j),g^{-1}(i)\rangle\langle i,j|\, , \end{equation}
is conjugated to the permutation operator $P$ as follows
\begin{equation}\label{eq:corr_loc}
\check r=F^{-1} PF=G^{-1} P G\,,
\end{equation}
where the operators $ F,G$ are defined by
\begin{equation}\label{eq:local_corr_gene}
	F=\sum_{i,j\in\llbracket 0,N-1\rrbracket}|i,g(j)\rangle\langle i,j|\ \ \ \ \ \text{and}\ \ \ \ \  G=\sum_{i,j\in\llbracket 0,N-1\rrbracket}|g^{-1}(i),j\rangle\langle i,j|\ .
\end{equation}
In addition to the operator $V$ introduced in \eqref{eq:defV}, we consider $U$, also  acting on the full lattice with $ L $ sites, and defined by
\begin{equation}\label{eq:defU}
    U=(g^{-1})^{L-1}\otimes (g^{-1})^{L-2} \otimes \dots\otimes (g^{-1}) \otimes \mathrm{id}\,.
\end{equation}
Using equation \eqref{eq:corr_loc} relating $ \rc $ and $ P$, we have
\begin{equation}\label{eq:corres}
	U\rc_{i,i+1}U^{-1}=V\rc_{i,i+1}V^{-1}=
	\begin{cases}
P_{i,i+1} &  \forall i\in \llbracket 1,L-1 \rrbracket\,,\\
 g_L^{L}\, P_{1,L}\,  g_L^{-L} \quad&  \mathrm{for}\ i=L\,,\ i+1=1\ ,
	\end{cases}			
	\end{equation}
where we recall that $g_L=\mathrm{id}^{\otimes (L-1)}\otimes g$.
The formula above shows immediately that we have
\begin{equation}\label{eq_conjugationMarkovproof}
	\tilde{M}_g=U^{-1}M_fU=V^{-1}M_fV\ \ \ \ \text{when $f=g^L$,}
\end{equation}
and this concludes the proof.
\end{proof}

\begin{example}
Let $g=(1^{(0)},1^{(1)})$ with the representation in the basis  $ \{ \lvert 1^{(0)} \rangle, \lvert 1^{(1)}\rangle \} $ 
\begin{equation*}
    g=\begin{pmatrix}
        \cdot & 1 \\
        1 & \cdot
    \end{pmatrix}=g^{-1}\,.
\end{equation*}
For a lattice with 2 sites, the bijections $V$ and $U$ are given by
\begin{equation*}
V=\mathrm{id}\otimes g =
\begin{pmatrix}
    \cdot & 1& \cdot & \cdot\\
    1 & \cdot& \cdot & \cdot\\
    \cdot & \cdot& \cdot & 1\\
    \cdot & \cdot& 1 & \cdot
\end{pmatrix}\,,\qquad U=g^{-1}\otimes \mathrm{id}=\begin{pmatrix}
    \cdot & \cdot& 1 & \cdot\\
    \cdot & \cdot& \cdot & 1\\
    1 & \cdot& \cdot & \cdot \\
    \cdot & 1& \cdot & \cdot
\end{pmatrix}\,.
\end{equation*}
\end{example}

An important feature of the operator $V$ conjugating the two Markov matrices in Formula \eqref{eq_conjugationMarkov} is 
that it is induced by a bijection of $\mathfrak{C}$ of the form  $\pi^{(1)}\pi^{(2)}\dots \pi^{(L)}\,,$
where each $\pi^{(i)}$ permutes the particles species at site $i$. Therefore, we conclude that the two Markov models given respectively by $\tilde{M}_g$ and $M_f$ (with $f=g^L$) are mathematically equivalent since an explicit bijection on the configuration space sends one into another. In other words, any mathematical property of one model is reflected to an equivalent property of the other model; for example, the number of sectors, their cardinality, etc. (see the next section).

\begin{remark}
One has to be careful that the bijection $V$ corresponds to local relabelings at each site $i$ of the lattice \emph{which depend on the site $i$}. This is not a simple relabeling of the set $\llbracket0,N-1\rrbracket$. Nevertheless, depending on the physical interpretation of the model, one may use the point of view corresponding to $\tilde{M}_g$ or the one corresponding to $M_f$.
\end{remark}

\begin{remark}
Notice that the existence of a permutation $g$ satisfying $ g^{L}=f $ for a given $ f $ is not guaranteed. A sufficient and necessary condition for the existence of an $ L $-th root of $ f $ is given in \cite{pouyanne2001}. Hence the set-theoretical models of the previous section are equivalent to a subclass (but not all) of the twisted SSEP models defined in this section.
\end{remark}

\section{Sectors and stationary states}\label{sec:comm_class_statio}

Having shown in the preceding section that the models from Section \ref{sec_Msettheo} constructed from set-theoretical solutions of the Yang--Baxter equation are mathematically equivalent to some twisted SSEP models from Section \ref{sec_MtwistedSSEP}, we proceed to the study of these twisted SSEP models. We will see that there are several stationary states which are all given by a uniform probability on sectors (irreducible subsets for $M_f$ of the configuration space). Therefore, we are reduced to the combinatorial study of these sectors.

In this section, it will be convenient to fix, as before, the bijection $f$ in the form (cycles decomposition)
\[f=(1^{(0)},\dots,1^{(c_1-1)})(2^{(0)},\dots,2^{(c_2-1)})\dots\dots(n^{(0)},\dots,n^{(c_n-1)})\ ,\]
and recall that the corresponding Markov matrix is denoted $M_f$.

We remind that the particles are denoted $s^{(e)}$ where $s$ is the species and $e$   its charge, with $1\leq s\leq n$ and $0\leq e\leq c_s-1$. The non-zero transition rates are given by
\begin{equation}
    \begin{split}
        & (s^{(e)}\,,\,t^{(e')}) \ \substack{1\\ \longleftrightarrow}\ (t^{(e')}\,,\,s^{(e)})\quad \mbox{for sites $(i,i+1)$ with $1\leq i<L$}\,,\\
        & (s^{(e)}\,,\,t^{(e')}) \ \substack{1\\ \longleftrightarrow}\ (t^{(e'+1)}\,,\,s^{(e-1)})\quad \mbox{between sites $L$ and 1}\,,
    \end{split}
\end{equation}
where the charges for the species $s$ are understood modulo $c_s$.

\subsection{Stationary states}\label{ssec:statio}
A stationary state $\lvert S \rangle=\sum_{\vec \tau\in \llbracket 0,N-1 \rrbracket^{L}} S(\vec \tau)\lvert \vec \tau\rangle $ is a probability vector 
 (the coefficients sum to $1$) which is invariant under the dynamics of the process, that is to say, it satisfies
$$ M_f\lvert S \rangle=0\ .$$
The number of stationary states is equal to the number of  sectors of the Markov process, which are defined as follows

\begin{definition}[sector]
For a Markov process defined on the configuration space $ \mathfrak{C} $ with transitions given by the Markov matrix $ M $, the sector of a configuration $ \vec \tau\in \mathfrak{C} $ corresponds to the equivalence class of $ \vec \tau $ for the following equivalence relation :
\begin{equation}
	\vec \tau \sim \vec \tau' \Leftrightarrow \exists k,l\in \mathbb{N},\ \langle \vec\tau \rvert M^{k} \lvert \vec\tau' \rangle\neq 0\ \mathrm{and}\ \langle \vec\tau' \rvert M^{l} \lvert \vec\tau \rangle\neq 0\, .
\end{equation}

\end{definition}

For all the processes considered in this article, the Markov matrices are symmetric so if there exists a $ k\in \mathbb{N} $ such that $ \langle \vec\tau \rvert M^{k} \lvert \vec\tau' \rangle\neq 0 $, then $ \langle \vec\tau' \rvert M^{k} \lvert \vec\tau \rangle\neq 0 $ and vice versa.

\begin{remark}\label{rem:group}
In the case of the twisted SSEP, it is easy to see from the definition of $M_f$ that the sector of a configuration $ \vec \tau $ is equivalent to the orbit of the group generated by the following two operations :
\begin{itemize}
\item Permuting the entries of $\vec{\tau}$;
\item the operation $f_{1,L}$ which increases the first charge by 1 and decreases the last charge by 1.
\end{itemize}
For the usual periodic SSEP, this group is the symmetric group. For the twisted periodic SSEP, it includes the symmetric group so the orbits are larger.

\end{remark}

\begin{remark}
If a Markov process has a single sector, the process is said to be irreducible.
This is not the case of the twisted periodic SSEP.
\end{remark}

We denote the set of sectors by
\begin{equation}
    \mathcal{S}=\{\text{sectors}\}=\{C_{\gamma}\,,\ \gamma\in\Gamma\}\, ,
\end{equation}
where $\Gamma$ is an indexing set for the sectors.

For any sector $ C_{\gamma} $, there exists a unique stationary state denoted $ \lvert S_{\gamma} \rangle $. This is proven in \cite[p.111,p118]{Norris1998} where the process is assumed to be positive recurrent on each class, a notion that we do not discuss further but which is true since the number of configurations in a class is finite \cite[p.27]{Norris1998}.
\begin{proposition}\label{prop:equiproba}
In each sector $C_{\gamma}$, the probability distribution of the stationary state is uniform, that is,
\begin{equation}\label{Sgamma}
	|S_{\gamma}\rangle=\frac{1}{\lvert C_{\gamma} \rvert}\sum_{\vec{\tau}\in C_\gamma}|\vec{\tau}\rangle\ .
\end{equation}
\end{proposition}

\begin{proof}   
We already know that the stationary state is unique for a given sector. Since (\ref{Sgamma}) is already normalised correctly (the coefficients sum to $1$), we only need to show that $M_f|S_{\gamma}\rangle=0$.

Recall that $M_f$ is symmetric and denote its matrix coefficients between two configurations by $m(\vec \tau' \rightarrow \vec \tau)$. We have
\begin{equation}
M_f\sum_{\vec{\tau}\in C_\gamma}|\vec{\tau}\rangle = \sum_{\vec \tau'\in C_{\gamma}}m(\vec \tau' \rightarrow \vec \tau) |\vec \tau'\rangle
=\sum_{\substack{\vec \tau'\in C_{\gamma}\\\vec \tau'\neq \vec \tau}}\left(m(\vec \tau' \rightarrow \vec \tau) -m(\vec \tau \rightarrow \vec \tau')\right)|\vec\tau\rangle=0\,,
\end{equation}
where we used Formula (\ref{coeffdiagM}) for the diagonal matrix coefficients of $M_f$.  
\end{proof}

\paragraph{The case of the usual $ N $-species SSEP.}
The usual $N$-species SSEP corresponds to the case $f=\text{Id}$. It is not irreducible and its sectors correspond to configurations with a fixed number of particles of each species. It is in bijection with the number of multisets of the size of the lattice $ L $ with elements taken in the set $\llbracket 0,N-1\rrbracket$. So the number of sectors is 
$$ |\mathcal{S}|=\begin{pmatrix}
	N-1+L\\
	N-1
\end{pmatrix}\ .$$
The cardinality of the sector $C_\gamma$ corresponding to $p_{i} $ particles of species $ i $ for $ i\in \llbracket 0,N-1 \rrbracket $ is the multinomial coefficient 
$$|C_{\gamma}|=\begin{pmatrix}
	L\\
	p_{0},...,p_{N-1}
\end{pmatrix}\ .$$
We are going to generalise these formulas to the general case (arbitrary $f$). Note that it is clear that the presence of the twist will reduce the number of sectors (while increasing their sizes), because the twist induces species transformations. Therefore, two different sectors of the usual multi-species SSEP can belong to the same sector of a twisted SSEP.

\begin{remark}
For any configuration $ \vec\tau,\vec\tau' $ of the twisted SSEP, the microscopic probability current $ j_{\vec\tau \rightarrow \vec\tau'}=m(\vec\tau \rightarrow \vec\tau')S(\vec\tau) - m(\vec\tau' \rightarrow \vec\tau)S(\vec\tau') $ vanishes for any stationary state $S$ and the system is at thermodynamic equilibrium as in the usual $ N $-species SSEP.
\end{remark}

\subsection{Invariants and sectors}\label{subsec_invcomm}

\subsubsection{Labelling sectors}\label{sssec:lab_comm}

Recall that a configuration is denoted by
\[\vec{\tau}=(s_1^{(e_1)},\dots,s_L^{(e_L)})\ \ \ \ \text{with}\ s_i\in\llbracket 1,n \rrbracket\ \text{and}\ e_i\in\llbracket 0,c_{s_i}-1 \rrbracket\, ,\]
where $s_i$ corresponds to the species and $e_i$ to the charge of the particle laying at site $i$.
\begin{definition}[Profile]\label{def:profile}
    The profile of a configuration $\vec{\tau}$ is defined by $p(\vec{\tau})=(p_1,\dots,p_n)$ where 
\[p_k=|\{i\in\{1,\dots,L\}\ \text{such that}\ s_i=k\}|\ \ \ \ \ \ \text{for}\ k=1,\dots,n\, .\]
In words, the value $p_k$ records how many times the species $k$ appears in $\vec{\tau}$.
\end{definition}

For a configuration $\vec\tau$, the species which appear at least once are called the \emph{$\vec\tau$-relevant species}. They depend only on the profile $p(\vec{\tau})$ and are the ones for which $p_k\neq 0$.
\begin{definition}[Total charge]\label{def:sum}
The total charge of a configuration $\vec{\tau}$ is defined by
\begin{equation}\label{sum}
    	E(\vec \tau)=\left(\sum_{i=1}^Le_i\right)mod\,D_{\vec\tau}\,,\ \ \ \text{where}\ D_{\vec\tau}=gcd(c_k\in K_{\vec \tau})\, ,
    \end{equation}
with $K_{\vec\tau}\subset \llbracket 1,n\rrbracket$ the set of $\vec\tau$-relevant species.
\end{definition}
Note that for a given position $i$, if in (\ref{sum}) we replace $e_i$ by, say, $e_i+c_{s_i}$, this does not change the value of the total charge, since by construction $D_{\vec\tau}$ divides each $c_{s_i}$ of $\vec\tau$-relevant species $s_i$. In this sense, the total charge is consistent with the definition of the charge of a species.

\begin{remark}
For $ K_{\vec\tau} $ the set of $\vec\tau$-relevant species, the total charge of this configuration is defined modulo $ D_{\vec \tau}=\mathrm{gcd}(c_{k}, k\in K_{\vec\tau}) $ which depends in general on $ \vec \tau $.
When the bijection $ f $ of the model $ M_{f} $ has all its cycles of the same size $ c $, then $ D_{\vec \tau}=c $ and it is independent on the configuration.
\end{remark}

We are ready to characterise the sectors.
\begin{proposition}\label{prop:lab_comm}
A sector is uniquely determined by a profile and a total charge. Namely, $\vec{\tau}$ and $\vec{\tau}'$ are in the same sector if and only if $p(\vec{\tau})=p(\vec{\tau}')$ and $E(\vec{\tau})=E(\vec{\tau}')$.
\end{proposition}
\begin{proof}

Recalling Remark \ref{rem:group}, the sector of a given $\vec{\tau}$ consists of all elements $\vec{\tau}'$ obtained by repeated applications of the permutation swapping the entries of $\vec \tau$ and of the permutation $f_{1,L}$ which increases the first charge by 1 and decreases the last charge by 1. Combining these two operations, one obtains also for any $i,j$ the operation $f_{i,j}$, which adds $1$ from the charge at position $i$ and subtracts $1$ to the charge at position $j$.

These operations clearly do not change the profile. For the total charge, let us consider the operation $f_{1,L}$. Then the charge $e_1$ is changed into $e_1-1$ or $e_1-1+c_{s_1}$ (if $e_1=0$) while the charge $e_L$ becomes $e_L+1$ or $e_L+1-c_{s_L}$ (if $e_L=c_{s_L}$). In any case, since $D_{\vec\tau}$ divides $c_{s_1}$ and $c_{s_L}$ the sum remains the same.

Now assume that $p(\vec{\tau})=p(\vec{\tau}')$ and $E(\vec{\tau})=E(\vec{\tau}')$. We need to show that $\vec{\tau}'$ is obtained from $\vec{\tau}$ by a succession of the above operations. Denote $\vec{\tau}=(s_1^{(e_1)},\dots,s_L^{(e_L)})$ as before. Since the profiles are the same, with permutations we can bring $\vec{\tau}'$ to be of the form
\[\vec{\tau}'=(s_1^{(e'_1)},\dots,s_L^{(e'_L)})\ .\]
Then using operations of the form $f_{i,L}$ for $i=1,\dots,L-1$, we can bring $\vec{\tau}'$ to be of the form
\[\vec{\tau}'=(s_1^{(e_1)},\dots,s_{L-1}^{(e_{L-1})},s_L^{(e'_L)})\ .\]
From the assumption that the total charges are equal, we have that $e'_L=e_L\,mod\,D_{\vec\tau}$. Hence, to conclude the proof, it is enough to show that we can change the value of $e'_L$ by $D_{\vec\tau}$.

Let $K_{\vec\tau}\subset\{1,\dots,n\}$ be the set of $\vec\tau$-relevant species. Recall that $D_{\vec\tau}=gcd(c_k\,,\ k\in K_{\vec\tau})$. Therefore, there exist integers $m_k\in\mathbb{Z}$ (Bezout identity) such that
\[\sum_{k\in K_{\vec\tau}}m_kc_k=D_{\vec\tau}\ .\]
Now for each $k\in K_{\vec\tau}$ such that $k\neq s_L$, we choose a position $i\in\{1,\dots,L-1\}$ such that $s_i=k$ and we apply the operation $f_{i,L}$ to the power $m_{s_i}c_{s_i}$. We see that the charge at position $i$ remains the same since we are adding a multiple of $c_{s_i}$. Having done that for every $k\in K_{\vec\tau}$ with $k\neq s_L$, we see that the charge at position $L$ has changed by
\[\sum_{k\in K_{\vec\tau}\backslash\{s_L\}}m_kc_k=D_{\vec\tau}-m_{s_L}c_{s_L}\ .\]
This is equivalent to changing the charge at position $L$ by $D_{\vec\tau}$ since this charge is considered modulo $c_{s_L}$.
\end{proof}

\subsection{Number of sectors}\label{ssec:comm_nb}

In Section \ref{ssec:statio}, we recalled the number of sectors for the usual $N$-species SSEP. Using the labelling of the sectors of the  section \ref{sssec:lab_comm}, we can now count them and generalise the formula giving their number for an arbitrary permutation $ f $.

\begin{proposition}\label{prop:comm_nb_gene}
The number of sectors $ \lvert \mathcal{S} \rvert $ of the twisted SSEP on a lattice with $ L\geq2 $ sites and with permutation $ f=(1^{(0)},\dots,1^{(c_1-1)})(2^{(0)},\dots,2^{(c_2-1)})\dots\dots(n^{(0)},\dots,n^{(c_n-1)}) $ is
\begin{equation}\label{eq:comm_nb_gene}
 \vert \mathcal{S}\rvert=\sum_{\emptyset\neq X\subset \llbracket 1,n \rrbracket}
 \begin{pmatrix}
 	L-1\\
	\lvert X \rvert-1
\end{pmatrix}
\mathrm{gcd}(c_{x},\ x\in X)\, .
\end{equation}
\end{proposition}

\begin{proof}
From Proposition \ref{prop:lab_comm}, the number of sectors $ \lvert\mathcal{S}\rvert $ is equal to the number of distinct pairs $ ( p(\vec \tau),E(\vec \tau)) $ for $ \vec \tau\in \llbracket 0,N-1 \rrbracket^{L} $.

For a given $ \emptyset \neq X\subset \llbracket 1,n \rrbracket $, we consider the configurations $\vec\tau$ such that $X$ is the set of $\vec \tau$-relevant species. By definition of the total charge, such configurations have $ \mathrm{gcd}( c_{x},\ x\in X) $ possible values for their total charge. The number of possible profiles is equal to the number of composition of $ L $ into $ \lvert X \rvert $ parts, which is $ 
\begin{pmatrix}
	L-1\\
	\lvert X\rvert-1
\end{pmatrix}$.
Thus there are $ 
\begin{pmatrix}
	L-1\\
	\lvert X\rvert-1
\end{pmatrix}\mathrm{gcd}( c_{x},\ x\in X) $ distinct pairs $ ( p(\vec \tau),E(\vec \tau)) $ for a given $ \emptyset \neq X\subset \llbracket 1,n \rrbracket $. The total number of sectors is obtained by summing over all nonempty subsets $ X\subset \llbracket 1,n \rrbracket $. 
\end{proof}

\begin{example}\label{ex:nb_com}
Here we apply formula \eqref{eq:comm_nb_gene} in the case where $ f $ has cycles of the same length, i.e., when all species have the same number of possible charges.

Let $ f=(1^{(0)},\dots,1^{(N-1)})^{J} $ where $ J\in \mathbb{N}$. The power $ J $ has the effect of splitting the full cycle into $ d $ cycles, each of length $ N/d $ where $ d=\mathrm{gcd}(N,J) $. In that case, formula \eqref{eq:comm_nb_gene} becomes

\begin{equation}
	\lvert\mathcal{S} \rvert=\sum_{\emptyset\neq X\subset \llbracket 1,d \rrbracket}
 \begin{pmatrix}
 	L-1\\
	\lvert X \rvert-1
\end{pmatrix}\frac{N}{d}=\frac{N}{d}\sum_{k=1}^{d}
\begin{pmatrix}
	d\\
	k
\end{pmatrix}						
 \begin{pmatrix}
 	L-1\\
	k-1
\end{pmatrix}=\frac{N}{d}\begin{pmatrix}			L+d-1\\						L
\end{pmatrix}\, ,
\end{equation}
where we grouped terms corresponding to sets $ X $ of size $ k $ in the second step and we used the Vandermonde's identity in the last one. 

The case $ f=\mathrm{Id} $ corresponds to $ J=0 $ so $ d=N $ and we recover the number of sectors of the multi-species periodic SSEP.

The case $ f=(1^{(0)},\dots,1^{(N-1)}) $ corresponds to $ J=1 $ so $ d=1 $ and there are $ N $ sectors corresponding to the $ N $ possible values of the total charge.
\end{example}

\begin{remark}\label{rem:increa}
When $ f $ has $d$ cycles of the same length $N/d$, the number of sectors seen as a function of $ d $, is strictly increasing. Indeed, for
\begin{equation}
	\mathcal{S}(d)=\frac{N}{d}\begin{pmatrix}			L+d-1\\						L
\end{pmatrix}\, ,
\end{equation}
and $ L\geq 2 $, we have that $ \frac{\mathcal{S}(d+1)}{\mathcal{S}(d)}=\frac{d+L}{d+1}>1 $. 
\end{remark}

\begin{remark}\label{rem:min_max}
The cases $ f=\mathrm{Id} $ and $ f=(1^{(0)},\dots,1^{(N-1)}) $ respectively maximise and minimise the number of sectors of the twisted SSEP. This is clear for $ f=\mathrm{Id} $ since a sector of the multi-species SSEP is always included in a (larger) sector of any twisted SSEP.

For $ f=(1^{(0)},\dots,1^{(c_1-1)})(2^{(0)},\dots,2^{(c_2-1)})\dots\dots(n^{(0)},\dots,n^{(c_n-1)}) $ we can rewrite Equation \eqref{eq:comm_nb_gene} as
\begin{align}
	\vert \mathcal{S}\rvert&=\sum_{k=1}^{n}
 \begin{pmatrix}
 	L-1\\
	1-1
\end{pmatrix}
c_{k}+\sum_{\substack{X\subset \llbracket 1,n \rrbracket\\ \lvert X\rvert>2}}
 \begin{pmatrix}
 	L-1\\
	\lvert X \rvert-1
\end{pmatrix}
\mathrm{gcd}(c_{x},\ x\in X)\\
			       &=N+\sum_{\substack{ X\subset \llbracket 1,n \rrbracket\\ \lvert X\rvert>2}}
 \begin{pmatrix}
 	L-1\\
	\lvert X \rvert-1
\end{pmatrix}
\mathrm{gcd}(c_{x},\ x\in X)\, .
\end{align}
The second term is positive so $ N $ is a lower bound for $ \lvert S \rvert $ which is reached for the full cycle $ f=(1^{(0)},\dots,1^{(N-1)})  $ so $ N $ is indeed the minimum of $ \lvert S\rvert $.
\end{remark}

\subsection{Cardinality of sectors}\label{ssec:comm_card}

We now determine the size of a sector and show that, for a fixed permutation $ f $, it only depends on the profile of the configurations of that sector. 
As in the previous section, we then apply the formula for a permutation with cycles of equal length.
\begin{proposition}\label{prop:comm_card_gene}
	
The cardinality of a sector $ C_{\gamma} $ with profile $ p=(p_{1},...,p_{n}) $ and total charge $E$ is 
\begin{equation}\label{eq:card_comm}
	\lvert C_{\gamma} \rvert = \left(\frac{L!}{p_{1}!...p_{n}!}\right)\frac{c_{1}^{p_{1}}... c_{n}^{p_{n}}}{\mathrm{gcd}( c_{x},\ x\in X  )}\, ,
\end{equation}
where $ X=\{k\in \llbracket 1,n \rrbracket\ \text{with}\ p_k\neq 0\}$ is the set of relevant species. In particular, the cardinality does not depend on the total charge $E$.
\end{proposition}

\begin{proof}
From Proposition \ref{prop:lab_comm}, a sector $ C_{\gamma} $ is uniquely labeled by a pair $ (p,E) $ where $p=(p_1,\dots,p_n)$ is the profile and $E$ is the total charge taken modulo $D$ with $D=\mathrm{gcd}(c_{x},x\in X) $. In the following, we first count the number of configurations with profile $ p$. Then, among those configurations, we count the ones that have total charge $ E $.

Ignoring the charges, the number of ways of distributing $ p_{1},\dots,p_{n} $ species of particles on a lattice with $ L $ sites is $\frac{L!}{p_{1}!...p_{n}!}$. Now taking into account the charge of particles, once the particle species is fixed for all sites, there are $ p_{s} $ particle of species $ s $ that can take $ c_{s} $ possible charges, for each $s$. This implies that the number of configurations with profile $ p $ is 
\begin{equation}\label{eq:card_pro}
	\frac{L!}{p_{1}!...p_{n}!}\prod_{s=1}^{n} c_{s}^{p_{s}}\, . 
\end{equation}
Now for a configuration $\vec{\tau}=(s_1^{(e_1)},\dots,s_L^{(e_L)})$, having the toal charge $E$ gives the following constraints on the charge:
\[e_1+\dots +e_L=E\ \mathrm{mod}\ D\ .\]
This only restricts the value of $e_L$ to be $E-(e_1+\dots+e_{L-1})\ \mathrm{mod}\ D$. Say the species in site $L$ is $k$ then $e_L$ is taken modulo $c_k$. Since $D$ divides $c_k$ by construction, there are $\frac{c_k}{D}$ possible values for $e_L$ instead of $c_k$ when the total charge is not fixed. Comparing with (\ref{eq:card_pro}), this proves (\ref{eq:card_comm}).

\end{proof}

\begin{example}\label{ex:comm_card}
We can apply formula \eqref{eq:card_comm} in the case where $ f $ has $d$ cycles of the same length $N/d$. We have
\begin{equation}\label{eq:card-sector}
	\lvert C_{\gamma} \rvert = \left(\frac{L!}{p_{1}!...p_{d}!}\right)\left(\frac{N}{d}\right)^{L-1} \, .
\end{equation}
The case $ f=\mathrm{Id} $ corresponds to $ d=N $ and we recover the multinomial coefficient corresponding to $ \lvert C_{\gamma}\rvert $ for the multi-species periodic SSEP. The case $ f=(1^{(0)},\dots,1^{(N-1)}) $ corresponds to $ d=1 $ and $ \lvert C_{\gamma}\rvert=N^{L-1} $. The power $ L-1 $ in the formula tells that the charge for the last site is fixed by the total charge.
\end{example}

\section{Twist tuning and branching probabilities}\label{sec:twist_branch}

In this part, we analyse the behaviour of a stationary state under a modification of the twist. Namely, we choose a stationary state for some twisted SSEP and we consider its evolution in a twisted SSEP with a different twist. In other words, we study a quench between two twisted SSEP with different twists. For example, when one of the models is the usual (non-twisted) SSEP, this procedure can be seen as turning on/off a twist on a given bond of the periodic lattice.

\subsection{Overlaps and branching probabilities}\label{ssec:overlap_branch}

We consider two bijections $f^{(1)}$ and $f^{(2)}$ of $\llbracket 0,N-1\rrbracket$, and we fix a stationary state $|S^{(1)}\rangle$ corresponding to a sector $C^{(1)}$ for the twisted SSEP associated to $f^{(1)}$. We recall that
\begin{equation}\label{probvector}|S^{(1)}\rangle=\frac{1}{|C^{(1)}|}\sum_{\vec\tau\in C^{(1)}}|\vec\tau\rangle\ .\end{equation}
Now we consider $|S^{(1)}\rangle$ as the initial condition for the Markov model corresponding to $f^{(2)}$, we let it evolve using the Markov matrix $M_{f^{(2)}}$ and we look at the possible stationary states (for $M_{f^{(2)}}$)
that are reached at $t=\infty$.

We denote $\mathcal{S}^{(2)}$ the set of sectors for the model associated to $f^{(2)}$, and by $|S^{(2)}_\gamma\rangle$ the stationary state associated to $\gamma\in\mathcal{S}^{(2)}$.
\begin{definition}
For $\gamma\in\mathcal{S}^{(2)}$, we denote by $prob(C^{(1)}\to C_{\gamma}^{(2)})$ the probability that the system ends up in the stationary state $|S_{\gamma}^{(2)}\rangle$. We call it the \emph{branching probability} from $C^{(1)}$ to $C_{\gamma}^{(2)}$.
\end{definition}
At $t=\infty$, the state will be of the form:
\begin{equation}\label{dec_alphagamma}\sum_{\gamma\in\mathcal{S}^{(2)}}prob(C^{(1)}\to C_{\gamma}^{(2)})|S^{(2)}_\gamma\rangle\ . \end{equation}
Note that we have
\[\sum_{\gamma\in\mathcal{S}^{(2)}}prob(C^{(1)}\to C_{\gamma}^{(2)})=1\,,\]
since the stationary vectors are normalised (in the sense that the sum of their components is 1, see \eqref{probvector}), and the evolution of a probability vector under a Markov process remains a probability vector.

The branching probabilities are simply calculated in terms of $|S^{(1)}\rangle$ and $|S_{\gamma}^{(2)}\rangle$. Indeed, since the stationary states $|S_{\gamma}^{(2)}\rangle$ do not evolve (by definition) under $M_{f^{(2)}}$, the coefficients in the expansion (\ref{dec_alphagamma}) are the same as the corresponding coefficients in the initial state $|S^{(1)}\rangle$. Therefore we have
\[prob(C^{(1)}\to C_{\gamma}^{(2)})=\frac{\langle S^{(1)}|S_{\gamma}^{(2)}\rangle}{\langle S_{\gamma}^{(2)}|S_{\gamma}^{(2)}\rangle}\ .\]
So far, the discussion was valid in a general setting. Now we use the explicit formula (\ref{probvector}) expressing that for our models, the stationary states correspond to uniform probabilities on sectors. We have immediately
\[\langle S^{(1)}|S_{\gamma}^{(2)}\rangle=\frac{|C^{(1)}\cap C_\gamma^{(2)}|}{|C^{(1)}||C_\gamma^{(2)}|}\ .\]
Therefore, noticing that $\langle S_{\gamma}^{(2)}|S_{\gamma}^{(2)}\rangle=\frac{1}{|C^{(2)}_\gamma|}$, we conclude with the following result.
\begin{proposition}\label{prop_branching prob}
The branching probabilities are given by:
\[prob(C^{(1)}\to C_{\gamma}^{(2)})=\frac{|C^{(1)}\cap C_\gamma^{(2)}|}{|C^{(1)}|}\ .\]
\end{proposition}

\subsection{Spreading, splitting and oscillations of sectors}\label{ssec:spread_split_osci}

\paragraph{Spreading of a sector.} Here we consider the particular case where the sector $C^{(1)}$ of the first model is included in a sector $C_{\gamma_0}^{(2)}$ of the second model, that is we have 
$$C^{(1)}\subset C_{\gamma_0}^{(2)}\,,\ \ \ \ \text{for some sector $\gamma_0\in\mathcal{S}^{(2)}$.}$$
In this case, the branching probabilities are trivial, in the sense that, according to Proposition \ref{prop_branching prob}, we have
\[prob(C^{(1)}\to C_{\gamma}^{(2)})=\delta_{\gamma,\gamma_0}\ .\]
The only thing to do is to determine, given $C^{(1)}$ and $C_{\gamma}^{(2)}$, whether $C^{(1)}\subset C_{\gamma}^{(2)}$ or $C^{(1)}\cap C_{\gamma}^{(2)}=\emptyset$. Examples are detailed in the next subsection.

\begin{remark}
Note that although the transition probability is trivial (0 or 1), the twist affects the dynamics of the twisted SSEP. Indeed, in the initial state, only the configurations $\vec\tau$ in $C^{(1)}$ have a non-zero probability $1/|C^{(1)}|$ while in the final state, more configurations $\vec\tau'\in C_{\gamma}^{(2)}$ have a non-zero (but smaller) probability
$1/|C_{\gamma}^{(2)}|$. 
\end{remark}

\paragraph{Splitting of a sector.} The opposite situation is when one reverses the role of $f^{(1)}$ and $f^{(2)}$: we start with the twist $f^{(2)}$ and we go on with the twist $f^{(1)}$. In this case, a sector $C^{(2)}$ for $f^{(2)}$ splits into a disjoint union of sectors for $f^{(1)}$:
\[C^{(2)}=C_{\gamma_1}^{(1)}\sqcup \dots \sqcup C_{\gamma_K}^{(1)}\ .\]
Now the branching probabilities are non-trivial and they are given, according to Proposition \ref{prop_branching prob} by
\[prob(C^{(2)}\to C_{\gamma}^{(1)})=\frac{|C_{\gamma}^{(1)}|}{|C^{(2)}|}\ \ \ \ \text{for $\gamma\in\{\gamma_1,\dots,\gamma_K\}$}\,,\]
and $0$ otherwise. Examples are given below.

\begin{remark}
The splitting phenomenon, where $C^{(2)}$ splits into a disjoint union of classes for $f^{(1)}$ appears systematically when $f^{(1)}$ is a power of $f^{(2)}$, namely if
\[f^{(1)}=(f^{(2)})^k\ \ \ \ \text{for some $k\geq 0$.}\]
Roughly speaking, the permutation $f^{(2)}$ permutes more things than the permutation $f^{(1)}$, therefore the sectors for $f^{(2)}$ are larger than the ones for $f^{(1)}$. 
\end{remark}

\paragraph{Oscillations of sectors.}
Assume we are in a spreading/splitting situation, namely we have a sector for $f^{(2)}$ which is a disjoint union of sector for $f^{(1)}$:
\[C^{(2)}=C_{\gamma_1}^{(1)}\sqcup \dots \sqcup C_{\gamma_K}^{(1)}\ .\]
Then we can start from a sector $C_{\gamma_{i_1}}^{(1)}$ then we alternate between turning on the twist $f^{(2)}$ and turning on the twist $f^{(1)}$. Then sector will alternatively spread and split, that is, schematically:
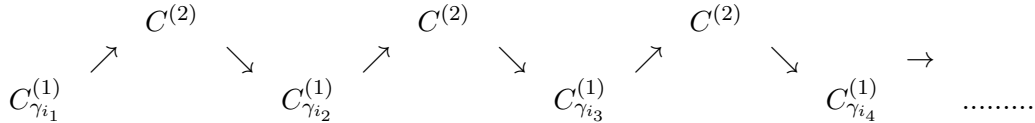
\[\begin{array}{ccccccccccccccc}
&&C^{(2)}& && &C^{(2)}& && &C^{(2)}& && &\\
&\nearrow &&\searrow &&\nearrow&&\searrow &&\nearrow &&\searrow &&\to &\\
C_{\gamma_{i_1}}^{(1)}& && &C_{\gamma_{i_2}}^{(1)}& && &C_{\gamma_{i_3}}^{(1)}& && &C_{\gamma_{i_4}}^{(1)}& &.........\\
\end{array}\]
where the up-arrow $\nearrow$ corresponds to switching on $f^{(2)}$ (and switching off $f^{(1)}$) and the down arrow to switching on $f^{(1)}$ (and switching off $f^{(2)}$).
We see that the configuration will oscillate between the sectors $C_{\gamma_1}^{(1)},\dots, C_{\gamma_K}^{(1)}$:
\[
C_{\gamma_{i_1}}^{(1)}\quad\raisebox{1ex}{$\substack{\scriptscriptstyle{prob(C^{(2)}\to C_{\gamma_{i_2}}^{(1)})}\\ \longrightarrow}$}\quad C_{\gamma_{i_2}}^{(1)}
\quad\raisebox{1ex}{$\substack{\scriptscriptstyle{prob(C^{(2)}\to C_{\gamma_{i_3}}^{(1)})}\\ \longrightarrow}$}\quad C_{\gamma_{i_3}}^{(1)}
\quad\raisebox{1ex}{$\substack{\scriptscriptstyle{prob(C^{(2)}\to C_{\gamma_{i_4}}^{(1)})}\\ \longrightarrow}$}\quad C_{\gamma_{i_4}}^{(1)}\to .........
\]

\begin{example}
Taking $f^{(1)}=\mathrm{Id}$, that is the usual multi-species SSEP, we see that turning on and off an arbitrary twist $f^{(2)}$, it is possible to move between different stationary states of the usual multi-species SSEP.
\end{example}

\subsection{Example: usual multi-species SSEP and an arbitrary twist}\label{ssec:ex1}

An obvious case where the spreading/splitting situation happens is when the first model is the usual multi-species SSEP (the twist $f^{(1)}=\mathrm{Id}$ is trivial) and the second model can have an arbitrary twist $f^{(2)}$.

In this case, we write $f^{(2)}$ with its cycle notation using species and charges, that is we take:
\[f^{(1)}=\mathrm{Id}\qquad\text{and}\qquad f^{(2)}=(1^{(0)},\dots,1^{(c_1-1)})(2^{(0)},\dots,2^{(c_2-1)})\dots\dots(n^{(0)},\dots,n^{(c_n-1)})\ .\]
A sector $C_{\gamma}^{(2)}$ for $f^{(2)}$ is indexed by a profile $p^{(2)}=(p_1,\dots,p_n)$ and a total charge $E^{(2)}$, as defined in Section \ref{subsec_invcomm}:
\[C_{p^{(2)},E^{(2)}}^{(2)}\quad :\qquad\quad p^{(2)}=(p_1,\dots,p_n)\ \ \ \text{and}\ \ \ E^{(2)}= k\ \mathrm{mod}\,D\,,\]
where $D$ is the greatest common divisor of the relevant species (see Definition \ref{def:sum}). Recall that the profile $p^{(2)}$ simply counts the number of occurrences of each species while the total charge $E^{(2)}$ is the sum of the charges modulo $D$.
Now a sector $C^{(1)}$ for the usual multi-species SSEP model is only indexed by a profile, which is of the form:
\[C_{p^{(1)}}^{(1)}\quad:\qquad\quad p^{(1)}=(p_{1^{(0)}},p_{1^{(1)}},\dots,p_{n^{(c_n-1)}})=\Bigl(p_{s^{(e)}}\Bigr)_{\substack{s=1,\dots,n\\ e=0,\dots,c_s-1}}\ .\]
Indeed recall that from the point of view of $f^{(1)}=\mathrm{Id}$, the profile simply counts the number of occurrences of every $s^{(e)}$, for various $s$ and $e$.

The inclusion condition is expressed as follows:
\begin{equation}\label{inclusion1}
C_{p^{(1)}}^{(1)}\subset C_{p^{(2)},E^{(2)}}^{(2)}\ \ \ \ \Longleftrightarrow\ \ \ \ \left\{\begin{array}{l}
\displaystyle\sum_{e=0}^{c_s-1}p_{s^{(e)}}=p_s\,,\ \ \forall s=1,\dots,n\,,\\[0.5em]
\displaystyle\sum_{s=1}^n\sum_{e=0}^{c_s-1}e p_{s^{(e)}}=E^{(2)}\ mod\,D\,.
\end{array}\right.
\end{equation}
Therefore we have according to Proposition \ref{prop_branching prob}
\[prob(C_{p^{(1)}}^{(1)}\to C_{p^{(2)},E^{(2)}}^{(2)})=\left\{\begin{array}{ll}
1 & \text{if (\ref{inclusion1}) is satisfied,}\\[0.4em]
0 & \text{otherwise.}
\end{array}\right.\]
This was the spreading situation. Now for the splitting situation, we reverse the position of the sectors, and we get, according to Proposition \ref{prop_branching prob}
\[prob(C_{p^{(2)},E^{(2)}}^{(2)}\to C_{p^{(1)}}^{(1)})=\left\{\begin{array}{ll}
\displaystyle\frac{|C_{p^{(1)}}^{(1)}|}{|C_{p^{(2)},E^{(2)}}^{(2)}|} & \text{if (\ref{inclusion1}) is satisfied,}\\[0.5em]
0 & \text{otherwise.}
\end{array}\right.\]
We recall that the explicit formula for the cardinality of a sector is given in Proposition \ref{prop:comm_card_gene}.

\subsection{Example: Two different non-trivial twists}\label{ssec:ex2}

Here, for the first model, we take for the twist the full cycle $f^{(1)}=(0,1,\dots,N-1)$. Therefore, in terms of charged species, for the first model, we have 
\begin{equation}
    f^{(1)}=(1^{(0)},1^{(1)},\dots,1^{(N-1)})\, ,
\end{equation}
and a sector $C^{(1)}_k$ is labeled by the total charge
$E^{(1)}= k\mbox{ mod}\ N$ (the profile being $p^{(1)}=(L)$). To get the charge $k$, one can put arbitrary charges in $L-1$ sites, the last one being used to fix the total charge. Hence, the sector has dimension $|C^{(1)}_k|=N^{L-1}$, in accordance with Example \ref{ex:comm_card}.

\paragraph{The case $f^{(2)}=(1,3,\dots N-1)(2,4,\dots N)$ with $N$ even.} For the second model, we take for $f^{(2)}$ the square of the first twist, which splits into two cycles since $N$ is even. In terms of charged species, we have
\begin{equation}
\begin{split}
    f^{(2)}=\ &(1^{(0)},1^{(2)},\dots,1^{(N-2)})(1^{(1)},1^{(3)},\dots,1^{(N-1)})\\
    =\ &(2^{(0)},2^{(1)},\dots,2^{(N/2-1)})(3^{(0)},3^{(1)},\dots,3^{(N/2-1)})\,.
\end{split}
\end{equation}
In the second line, we have written the charged species from the point of view of $f^{(2)}$. Since $f^{(2)}$ only permutes the elements $1^{(e)}$ for a given parity of $e$, it is as if we had two different species. In other words, we have set $2^{(e)}=1^{(2e)}$ and $3^{(e)}=1^{(2e+1)}$.

A sector for $f^{(2)}$ is labeled by the profile $p^{(2)}=(p_2,p_3)$ (with $p_2+p_3=L$) and the total charge
$E^{(2)}= \ell\mbox{ mod}\ N/2$. A particular configuration for this sector is given by
\begin{equation}
    \begin{split}
  &  (\underbrace{2^{(\ell)},2^{(0)},\dots,2^{(0)}}_{p_2},\underbrace{3^{(0)},3^{(0)},\dots,3^{(0)}}_{p_3})\quad \mbox{if } p_2\neq0\,,\quad
   (\underbrace{3^{(\ell)},3^{(0)},\dots,3^{(0)}}_{p_3=L})\quad \mbox{if } p_2=0\,.
    \end{split}
\end{equation} 
In both cases, the charge $E^{(1)}$ for this configuration is $E^{(1)}=2\ell+p_3$ mod $N$. Since $p_2+p_3=L$ is always satisfied, $prob(C^{(1)}_k\to C^{(2)}_{p^{(2)},\ell})$ is non-vanishing if and only if
$k=2\ell+p_3$ mod $N$. Denoting by $C^{(2)}_\gamma$ the sectors which obey this equality, the probability is given by
\begin{equation}
    prob(C^{(1)}_k\to C^{(2)}_{\gamma})=\frac{|C^{(2)}_{\gamma}|}{|C^{(1)}_k|}=\frac{L!}{p_2!\,p_3!}\left(\frac12\right)^{L-1}\,.
\end{equation}
\begin{remark}\label{rem_checkprob}
Using this expression, one can check that $\sum_{\gamma}prob(C^{(1)}_k\to C^{(2)}_{\gamma})=1$, as it should. Note that the sum should be taken over $l$ and $(p_2,p_3)$ with the constraints that $p_2+p_3=L$ and $p_3=k-2l$ mod $N$. Depending on $k$, only the even or only the odd values of $p_3$ are allowed. Then one uses the equality
\begin{equation}
    \sum_{\substack{p_2+p_3=L\\ p_3= a \mbox{ mod }2}}\frac{L!}{p_2! p_3!} = 2^{L-1}\,,
\end{equation}
valid for $a=0,1$ (this is checked by expanding $(1+1)^L$ and $(1-1)^L$).
\end{remark}

\paragraph{The case $f^{(2)}=(f^{(1)})^{n-1}$ with $N=(n-1)D$.} More generally, keeping $f^{(1)}$ as above, we assume now that $N=(n-1)D$ (the previous case was $n=3$). We take 
\begin{equation}
    f^{(2)}=(2^{(0)},...,2^{(D-1)}) (3^{(0)},...,3^{(D-1)}) \cdots (n^{(0)},...,n^{(D-1)})\,,
\end{equation} 
with the reindexing done similarly as before. 

The sectors are now characterised by the profile $p^{(2)}=(p_2,p_3,\dots,p_{n})$ and a total charge $E^{(2)}= \ell\mbox{ mod}\ D$.
A particular configuration for this sector is given by
\begin{equation}
  (\underbrace{2^{(\ell)},2^{(0)},\dots,2^{(0)}}_{p_2},\underbrace{3^{(0)},3^{(0)},\dots,3^{(0)}}_{p_3},\dots,\underbrace{n^{(0)},n^{(0)},\dots,n^{(0)}}_{p_n})\,,
\end{equation} 
where we have supposed that $p_2\neq0$ without loss of generality. The condition of non-empty intersection between $C^{(1)}_k$ and $C^{(2)}_{p^{(2)},\ell}$ is expressed as
\begin{equation}
    (n-1)\,\ell+\sum_{j=2}^n (j-2)p_j = k\mbox{ mod }N\, ,
\end{equation}
and, when this condition is satisfied,
\begin{equation}
    prob(C^{(1)}_k\to C^{(2)}_{\gamma})= \frac{L!}{p_2! p_3!\dots p_n!}\,\left(\frac1{n-1}\right)^{L-1}\,.
\end{equation}

\begin{remark}
As in Remark \ref{rem_checkprob}, the property $\sum_{\gamma}prob(C^{(1)}_k\to C^{(2)}_{\gamma})=1$ can be checked directly and is ensured by the identity
\begin{equation}
    \sum_{\substack{p_2+...+p_n=L\\ p_3+2p_4+..+(n-2)p_n= a \mbox{ mod }n-1}}\frac{L!}{p_2! p_3!\dots p_n!} = (n-1)^{L-1}\,,
\end{equation}
which is valid for all values $a=0,1,...,n-2$. This is checked by expanding $(1+\xi+\dots+\xi^{n-2})^L$ for all $(n-1)$-st roots of unity $\xi$.
\end{remark}

\section{More general solutions of set-theoretical Yang-Baxter equation}\label{sec:set_theo_gene}

In this section, we introduce Markov models constructed from more general set-theoretical maps than the ones previously considered. We then give a non trivial example of such map. We show that its corresponding Markov model for a small lattice ($L=3$) is not equivalent to any twisted periodic SSEP coming from Lyubashenko solutions.

\subsection{Markov models}\label{ssec:mark_set_gene}

We now consider set-theoretical maps that are more general than the ones in Section \ref{ssec:lyub_sol}. For $ i\in \llbracket 0,N-1 \rrbracket $, we consider the family of bijective maps
\[g_{i},\ f_{i}:\ \llbracket 0,N-1\rrbracket\to\llbracket 0,N-1\rrbracket \, ,\]
that are chosen in such a way that the following operator on two sites
\begin{equation}\label{eq:r_gene}
\check r=\sum_{i,j\in \llbracket 0,N-1\rrbracket}|g_{i}(j),f_{j}(i)\rangle\langle i,j|\ .
\end{equation}
satisfies the YBE \eqref{eq:cYB} and is involutive as in equation \eqref{eq:invo}.

\begin{remark}
A Lyubashenko solution has the same form as Equation \eqref{eq:r_gene}, but the bijections $ g_{i} $ and $ f_{i} $ are independent of $ i $, so that they are respectively all equal to some $ g $ and $ f $.
\end{remark}

The maps $ g_{i} $ and $ f_{i} $ such that $ \rc $ obeys the YBE, satisfy for all $ i,j,k\in \llbracket 0,N-1 \rrbracket $
\begin{equation}\label{eq:yb_rel}
	g_{i}(g_{j})(k)=g_{g_{i}(j)}(g_{f_{j}(i)}(k)),\quad f_{k}(f_{j}(i))=f_{f_{k}(j)}(f_{g_{j}(k)}(i)),\quad f_{g_{f_{j}(i)}(k)}(g_{i}(j))=g_{f_{g_{j}(k)}(i)}(f_{k}(j))\,.
\end{equation}
The maps $ g_{i} $ and $ f_{i} $ such that $ \rc^{2}=\mathrm{Id} $ satisfy for all $ i,j\in \llbracket 0,N-1 \rrbracket $
\begin{equation}\label{eq:inv_rel}
	g_{g_{i}(j)}(f_{j}(i))=i,\quad f_{f_{j}(i)}(g_{i}(j))=j\,.
\end{equation}
Involutive solutions of the YBE where $ g_{i} $, $ f_{i} $ are bijections are equivalent to algebraic structures called braces \cite{Rump2007}.

The local jump operator and the Markov matrix constructed from this operator are defined identically as previously in Equations \eqref{eq:loc_jump} and \eqref{eq:M_Lyub}. It can be easily checked that the resulting model is a Markovian periodic model. It also has the property of being integrable. This comes from the fact that the operator \eqref{eq:r_gene} is involutive, which allows to Baxterize it and to construct a set of commuting operators as in Section \ref{ssec:int}.

\paragraph{Interpretation.}
The form of the operator \eqref{eq:r_gene} implies that the transition rates between two adjacent sites have the form
\begin{equation}\label{localprocess_gene}
	(\,\tau_{i}\,,\,\tau_{i+1}\,)\ \ \to\ \ (\,g_{\tau_{i}}(\tau_{i+1})\,,\,f_{\tau_{i+1}}(\tau_{i})\,)\  \  \ \ i=1,\dots,L\, .
\end{equation}
The interpretations that we can give for models constructed from those more general set-theoretical solutions are in general harder to find than the ones given in section \ref{ssec:interpret}.

\paragraph{Separable and non-separable bijections. } In Section \ref{ssec:corr_set}, Proposition \ref{prop:corr_set_twist}, we proved the equivalence between Markov models constructed from Lyubashenko solutions and the twisted periodic SSEP. In that case, the bijection on the configurations had the form
\begin{equation}\label{eq:V_sep}
	V_{\mathrm{sep}}=\pi^{(1)} \pi^{(2)}... \pi^{(L)}\, ,
\end{equation}
and we call it separable. For set-theoretical solutions \eqref{eq:r_gene}, there exist non separable bijections of $ \mathfrak{C} $ given by 
\begin{equation}\label{Ufactor_gene}
U(\tau_1,\tau_2,\dots,\tau_L)=(f_{\tau_{L}}f_{\tau_{L-1}}\dots f_{\tau_{2}}(\tau_{1}),\dots,f_{\tau_{L}}(\tau_{L-1}),\tau_{L})\, ,
\end{equation}
and
\begin{equation}\label{Vfactor_gene}
V(\tau_1,\tau_2,\dots,\tau_L)=(\tau_1,g_{\tau_{1}}(\tau_2),\dots,g_{\tau_{1}}g_{\tau_{2}}\dots g_{\tau_{L-1}}(\tau_L))\,,
\end{equation}
satisfying 

\begin{equation}\label{eq:loc_corr_UV}
	\forall i\in \llbracket 1,L-1 \rrbracket,\ \rc_{i,i+1}=U^{-1}P_{i,i+1} U=V^{-1}P_{i,i+1} V\, .
\end{equation}

The proof of equation \eqref{eq:loc_corr_UV} for  $ U $ is done by induction on $ L $ in \cite[prop1.7,p7]{etingof1998} and is similar for $ V $. However, in addition to the fact that $U$ and $V$ are not separable, the conjugation of the model constructed from $ \rc $  by $ U $ or $ V $ does not give in general a twisted periodic SSEP. This is because $ U\rc_{L,1}U^{-1} $ (or $ V\rc_{L,1}V^{-1} $) acts on the whole lattice in general and not only on sited $L$ and $1$.

\begin{remark}
If a set-theoretical solution \eqref{eq:r_gene} is conjugated to the permutation operator via a separable bijection of configurations, then it is a Lyubashenko solution.
Indeed, for two sites suppose we have
\begin{equation}
	\rc_{1,2}=V_{\mathrm{sep}}^{-1}P_{1,2}V_{\mathrm{sep}}\, ,
\end{equation}
with $ V_{\mathrm{sep}}=\pi^{(1)} \pi^{(2)} $. On one side, we have for $ u,v\in \llbracket 0,N-1 \rrbracket $
\begin{equation}
	\rc_{1,2}(u,v)=(g_{u}(v),f_{v}(u))\, .
\end{equation}
On the other side, we have
\begin{equation}
	V_{\mathrm{sep}}^{-1}P_{1,2}V_{\mathrm{sep}}(u,v)=((\pi^{(1)})^{-1} \pi^{(2)}(v), (\pi^{(2)})^{-1} \pi^{(1)}(u))\, .
\end{equation}
Since $ u $ and $ v $ are arbitrary, by equating each component, this implies that $ g_{u}=g $ for all $ u $ and $ f_{v}=f $ for all $ v $. Using the involutivity of $ \rc $, this implies that $ f=g^{-1} $ so it has the form of Equation \eqref{eq:r_lyub1}.

\end{remark}

\subsection{Example}\label{ssec:non_eq}

We now give an example of a set-theoretical solution of the form of Equation \eqref{eq:r_gene} that is not a Lyubashenko solution. We then show that the corresponding Markov model is not equivalent to any twisted SSEP coming from Lyubashenko solutions for $L=3$.

This example is for $ N=3 $.
The maps $ f_{i} $ and $ g_{j} $ of Equation \eqref{eq:r_gene} are given by
\begin{align*}
	g_{0}(0)=f_{0}(0)=2,\quad g_{0}(1)=f_{0}(1)=1,\quad g_{0}(2)=f_{0}(2)=0\\
	g_{1}(0)=f_{1}(0)=0,\quad g_{1}(1)=f_{1}(1)=1,\quad g_{1}(2)=f_{1}(2)=2\\
	g_{2}(0)=f_{2}(0)=2,\quad g_{2}(1)=f_{2}(1)=1,\quad g_{2}(2)=f_{2}(2)=0\, .
\end{align*}
They are indeed bijective and satisfy equations \eqref{eq:yb_rel} and \eqref{eq:inv_rel}.

In the basis $ \mathcal{B}=\{ \lvert 00 \rangle, \lvert 01 \rangle, \lvert 02 \rangle, \lvert 10 \rangle, \lvert 11 \rangle, \lvert 12 \rangle, \lvert 20 \rangle, \lvert 21 \rangle, \lvert 22 \rangle\} $, the matrix  $ \rc $ reads 
\begin{equation}\label{eq:r_brace}
	\rc=
\begin{pmatrix}
	\cdot &\cdot&\cdot &\cdot&\cdot &\cdot&\cdot &\cdot&1\\
\cdot &\cdot&\cdot &1&\cdot &\cdot&\cdot &\cdot&\cdot\\
\cdot &\cdot&1 &\cdot&\cdot &\cdot&\cdot &\cdot&\cdot\\
\cdot &1&\cdot &\cdot&\cdot &\cdot&\cdot &\cdot&\cdot\\
\cdot &\cdot&\cdot &\cdot&1 &\cdot&\cdot &\cdot&\cdot\\
\cdot &\cdot&\cdot &\cdot&\cdot &\cdot&\cdot &1&\cdot\\
\cdot &\cdot&\cdot &\cdot&\cdot &\cdot&1 &\cdot&\cdot\\
\cdot &\cdot&\cdot &\cdot&\cdot &1&\cdot &\cdot&\cdot\\
1 &\cdot&\cdot &\cdot&\cdot &\cdot&\cdot &\cdot&\cdot\\
\end{pmatrix}\, ,
\end{equation}
where the dots stand $0$. The Markov matrix of the model constructed from \eqref{eq:r_brace} is denoted $ M_{\rc} $ and we denote $ M_{f} $ the Markov matrix of the twisted SSEP with the twist defined by $f$.

We are going to show that the matrix $ M_{\rc} $ is not conjugated to any matrix $ M_{f} $ for $L=3$ (and in particular, there is no bijection of the configuration space, separable or not, which sends $ M_{\rc} $ to a matrix $ M_{f} $).
To do that, we show that the number of sectors for $M_{\check r}$ is different from the number of sectors for $M_f$ for any permutation $ f $. This will show that those models are not conjugated since the number of sectors of a model is equal to the dimension of the kernel of the corresponding Markov matrix, and this is invariant by conjugation.

For $ N=3 $, there are $ 3!=6 $ different permutations $ f $. As we have seen in Section \ref{sec:comm_class_statio}, the number of sectors only depends on the cycle structure of $ f $.
Using the notation introduced in Section \ref{ssec:interpret}, it suffices to consider $ f_{1}=\mathrm{Id} $, $ f_{2}=(1^{(0)},1^{(1)})(2^{(0)}) $ and $ f_{3}=(1^{(0)},1^{(1)},1^{(2)}) $. Using Equation \eqref{prop:comm_nb_gene}, the number of sectors in the twisted SSEP corresponding to $ f_{1} $, $ f_{2} $ and $ f_{3} $ for $L=3$ are respectively 10, 5 and 3.

On the other hand, by considering the action of \eqref{eq:r_brace} on the $3^3$ configurations, we find that the number of sectors in the corresponding set-theoretical model constructed for $L=3$ is 7. This proves that this set-theoretical model is not conjugated to any twisted SSEP coming from Lyubashenko solutions for $L=3$.

\begin{remark}
This model will be studied in more details  in a future work. 
In particular, we will show that, for all $L>2$, it is not equivalent to the previous twisted SSEP with 3 species.
\end{remark}

\section{Outlook}

The main focus of this paper was the construction and study of Markov models from the simple Lyubashenko solutions of the YBE.
We also introduced Markov models constructed from more general set-theoretical solutions of the YBE that are not always equivalent to some twisted SSEP and we gave an example of such solution. In a future work, we will study in more details the Markov model constructed from this solution.

Additionally, the oscillation phenomenon described in Section \ref{ssec:spread_split_osci} could be further investigated. 
On the mathematical side, given two twisted SSEP, what are the conditions on the permutations defining the twists in order to have a non-zero probability of occupying any configuration after a certain number of switch between the twists of both models?
On the condensed matter physics side, we can imagine that the modifications of the twist are caused by an unstable field on a bond of the lattice and affecting the charges of particles. We could study the finite time effects of this field on the system for different types of instabilities (periodic, transient, etc.). A similar analysis can be carried in the original set-theoretical model where the charge of particles can vary on each bond of the lattice when exchanging (because of unstable chemical bonds for instance).

We focused throughout this article on models at thermodynamic equilibrium defined on a ring. As a natural continuation of this work, we plan to extend our approach to broader classes of models, with particular emphasis on  those leading to out-of-equilibrium systems.

A first direction consists in adding integrable boundaries to the system. The framework for such construction has been developed in \cite{Caudrelier2013,Smoktunowicz2020,Doikou_2021_2}. It led to the notion of set-theoretical reflection equation \cite{Caudrelier2012, Caudrelier2013}, whose solutions  provide the possible integrable boundaries.

A second approach preserves periodic boundary conditions but the set-theoretical solution is deformed so that the local jump operator constructed from it has a parameter dependence leading to an asymmetry in the transition probabilities.

Clearly, it will be highly desirable to combine these two approaches, namely by studying solutions of the set-theoretic reflection equation in conjunction with parameter-dependent set-theoretic-like solutions of the YBE.

\paragraph{Acknowledgements.} We acknowledge support from the grant IRP AAPT. We warmly thank Nicolas Cramp\'e for helpful discussions. M. D. also thanks the Laboratoire de Mathématiques de Reims for its hospitality during a research visit related to this project.


\begin{thebibliography}{99}

\bibitem{Alcaraz1994} F. Alcaraz, M. Droz, M. Henkel and V. Rittenberg, \textsl{Reaction-Diffusion Processes, Critical Dynamics, and Quantum Chains}, \href{https://doi.org/10.1006/aphy.1994.1026}{Ann. Physics \textbf{230}, 250--302 (1994)}, \texttt{arXiv:hep-th/9302112}.

\bibitem{Bachiller2016}D. Bachiller, \textsl{Solutions of the Yang–Baxter equation associated to skew left braces, with applications to racks}, \href{https://doi.org/10.1142/S0218216518500554}{J. Knot Theory Ramifications \textbf{27}, 1850055 (2016)}, \texttt{arXiv:1611.08138}. 

\bibitem{Baxter1982Exactly} R. Baxter, \textsl{Exactly Solved Models in Statistical Mechanics}, Academic Press (1982).

\bibitem{Batchelor1995} M. Batchelor, R. Baxter, M. O'Rourke and C. Yung, \textsl{Exact solution and interfacial tension of the six-vertex model with anti-periodic boundary conditions}, \href{https://doi.org/10.1088/0305-4470/28/10/009}{J. Phys. A: Math. Gen. \textbf{28}, 2759 (1995)}, \texttt{arXiv:hep-th/9502040}.

\bibitem{cao2013} J. Cao, W. Yang, K. Shi and Y. Wang, \textsl{Off-Diagonal Bethe Ansatz and Exact Solution of a Topological Spin Ring}, \href{https://doi.org/10.1103/PhysRevLett.111.137201}{Phys. Rev. Lett. \textbf{111}, 137201 (2013)}, \texttt{arXiv:1305.7328}.

\bibitem{Caudrelier2012}V. Caudrelier and Q. Zhang, \textsl{Yang-Baxter and reflection maps from vector solitons with a boundary}, \href{https://doi.org/10.1088/0951-7715/27/6/1081}{Nonlinearity \textbf{27}, 1081--1103 (2014)}, \texttt{arXiv:1205.1133}.

\bibitem{Caudrelier2013}V. Caudrelier, N. Crampé, and Q. Zhang, \textsl{Set-theoretical reflection equation: classification of reflection maps}, \href{https://doi.org/10.1088/1751-8113/46/9/095203}{J. Phys. A: Math. Theor. \textbf{46}, 095203 (2013)}, \texttt{arXiv:1210.5107}.

\bibitem{cedo2008} F. Cedó, E. Jespers and A. Rio, \textsl{Involutive Yang-Baxter Groups}, \href{https://doi.org/10.1090/S0002-9947-09-04927-7 }{Trans. Amer. Math. Soc. \textbf{362} (2010), 2541--2558}, \texttt{arXiv:0803.4054}.

\bibitem{Cedo2012}F. Cedó, E. Jespers and J. Okniński, \textsl{Braces and the Yang–Baxter Equation}, \href{https://10.1007/S00220-014-1935-Y}{Comm. Math. Phys. \textbf{327} 101--116 (2012)}.

\bibitem{Chouraqui2009} F. Chouraqui, \textsl{Garside Groups and Yang–Baxter Equation}, \href{http://dx.doi.org/10.1080/00927870903386502}{Comm. Algebra \textbf{38}, 4441--4460 (2010)}, \texttt{arXiv:0912.4827}.

\bibitem{Doikou_2021} A. Doikou, \textsl{Set-theoretic Yang–Baxter equation, braces and Drinfeld twists}, \href{http://dx.doi.org/10.1088/1751-8121/ac219e}{J. Phys. A: Math. Theor. \textbf{54}, 415201 (2021)}, \texttt{arXiv:2102.13591}.

\bibitem{Doikou_2021_2} A. Doikou and A. Smoktunowicz, \textsl{Set-theoretic Yang–Baxter and reflection equations and quantum group symmetries}, \href{https://doi.org/10.1007/s11005-021-01437-7}{Lett.  Math. Phys. \textbf{111}, 105 (2021)}, \texttt{arXiv:2003.08317}.

\bibitem{Drinfeld1992} V. Drinfeld, \textsl{On some unsolved problems in quantum group theory}, \href{https://doi.org/10.1007/BFb0101175}{Quantum Groups Lect. Notes Math. \textbf{1510}, 1--8 (1992)}.

\bibitem{etingof1998} P. Etingof, T. Schedler and A. Soloviev, \textsl{Set-theoretical solutions to the quantum Yang-Baxter equation}, \href{https://doi.org/10.1215/S0012-7094-99-10007-X}{Duke Math. J. \textbf{100}(2), 169--209 (1999)}, \texttt{arXiv:math/9801047}.

\bibitem{Evans1995}M. Evans, D. Foster, C. Godrèche and D. Mukamel, \textsl{Asymmetric exclusion model with two species: Spontaneous symmetry breaking}, \href{https://doi.org/10.1007/BF02178354}{J. Stat. Phys. \textbf{80}, 69--102 (1995)}.

\bibitem{faddeev1996}
L.D. Faddeev, \textsl{How algebraic Bethe ansatz works for integrable model}, Les Houches Lectures Quantum
Symmetries de A. Connes et al. (Amsterdam: North-Holland, 1998) p. 149, \texttt{arxiv:hep-th/9605187}.

\bibitem{FST}
L.D. Faddeev, E.K. Sklyanin and L.A. Takhtajan, \textsl{Quantum inverse problem method. I}, \href{https://doi.org/10.1007/BF01018718}{Theor. Math. Phys. \textbf{40}, 688--706 (1979)}.

\bibitem{Gateva1998} T. Gateva-Ivanova and M. Van den Bergh, \textsl{Semigroups of I-Type}, \href{https://doi.org/10.1006/jabr.1997.7399}{J. Algebra \textbf{206}, 97--112 (1998)}, \texttt{arXiv:math/0308071}.

\bibitem{Gateva2004} T. Gateva-Ivanova, \textsl{A combinatorial approach to the set-theoretic solutions of the Yang–Baxter equation}, \href{https://doi.org/10.1063/1.1788848}{J. Math. Phys. \textbf{45}, 3828--3858 (2004)}, \texttt{arXiv:math/0404461}.

\bibitem{Gier2005}J. Gier and F. Essler, \textsl{Bethe Ansatz Solution of the Asymmetric Exclusion Process with Open Boundaries}, \href{https://doi.org/10.1103/PhysRevLett.95.240601}{Phys. Rev. Lett. \textbf{95}, 240601 (2005)}, \texttt{cond-mat/0508707}.

\bibitem{Golinelli2004}O. Golinelli and K. Mallick, \textsl{Bethe ansatz calculation of the spectral gap of the asymmetric exclusion process}, \href{https://doi.org/10.1088/0305-4470/37/10/001 }{J. Phys. A Math. Gen. \textbf{37}, 3321 (2004)}, \texttt{arXiv:cond-mat/0312371}.

\bibitem{Guarnieri2015}L. Guarnieri and L. Vendramin, \textsl{Skew braces and the Yang-Baxter equation}, \href{https://doi.org/10.1090/mcom/3161}{Math. Comput. \textbf{86} 2519--2534 (2015)}, \texttt{arXiv:1511.03171}.

\bibitem{Hao2016} K. Hao, J. Cao, G. Li, W. Yang, K. Shi and Y. Wang, \textsl{Exact solution of an su(n) spin torus}, \href{https://doi.org/10.1088/1742-5468/2016/07/073104}{J. Stat. Mech.: Theory Exp. \textbf{2016}, 073104 (2016)}, \texttt{arXiv:1601.04389}.

\bibitem{Isaev2004} A. Isaev, \textsl{Quantum Groups and Yang–Baxter Equations}, \href{https://archive.mpim-bonn.mpg.de/id/eprint/3447/3/preprint_2004_132.pdf}{Max-Planck-Institut für Mathematik Preprint (2004)}.

\bibitem{Jones1990} V. F. R. Jones, \textsl{Baxterization}, \href{https://doi.org/10.1007/978-1-4684-9148-7_2}{Int. J. Mod. Phys. B \textbf{4}, 701--713 (1990)}.

\bibitem{Karimipour1999}V. Karimipour, \textsl{ Multispecies asymmetric simple exclusion process and its relation to traffic flow}, \href{https://doi.org/10.1103/PhysRevE.59.205}{Phys. Rev. E \textbf{59}, 205--212 (1999)}. 

\bibitem{Kawa1966}K. Kawasaki, \textsl{Diffusion Constants near the Critical Point for Time-Dependent Ising Models. I.}, \href{https://doi.org/10.1103/PhysRev.145.224}{Phys. Rev. \textbf{145}, 224--230 (1966)}.

\bibitem{lee1952}T. Lee and C. Yang, \textsl{Statistical Theory of Equations of State and Phase Transitions. II. Lattice Gas and Ising Model}, \href{https://doi.org/10.1103/PhysRev.87.410}{Phys. Rev. \textbf{87}, 410--419 (1952)}.

\bibitem{mac1968}C. MacDonald, J. Gibbs and A. Pipkin,  \textsl{Kinetics of biopolymerization on nucleic acid templates}, \href{https://doi.org/10.1002/bip.1968.360060102}{Biopolymers \textbf{6}, 1--25 (1968)}.

\bibitem{Kirone2015} K. Mallick, \textsl{The exclusion process: A paradigm for non-equilibrium behaviour}, \href{https://cea.hal.science/cea-01307568/document}{Phys. A: Stat. Mech. Appl. \textbf{418}, 17--48 (2015)}, \texttt{arXiv:1412.6258}.

\bibitem{Norris1998}J. Norris, \textsl{Markov Chains}, \href{https://doi.org/10.1017/CBO9780511810633}{Cambridge University Press (2012)}.

\bibitem{pouyanne2001} N. Pouyanne, \textsl{On the Number of Permutations Admitting an m-th Root}, \href{https://doi.org/10.37236/1620}{Electron. J. Comb. \textbf{9} (2002)}.

\bibitem{Rump2005}W. Rump, \textsl{A decomposition theorem for square-free unitary solutions of the quantum Yang-Baxter equation}, \href{https://doi.org/10.1016/j.aim.2004.03.019}{	Adv. Math. \textbf{193}, 40--55 (2005)}.

\bibitem{Rump2007}W. Rump, \textsl{Braces, radical rings, and the quantum Yang–Baxter equation}, \href{https://doi.org/10.1016/j.jalgebra.2006.03.040}{J. Algebra \textbf{307}, 153--170 (2007)}.

\bibitem{Schutz2001} G. Schütz, \textsl{Exactly Solvable Models for Many-Body Systems Far from Equilibrium}, \href{https://doi.org/10.1016/S1062-7901(01)80015-X}{Ph. Transit. Crit. Phenom. \textbf{19}, 1 (2001)}.

\bibitem{Smoktunowicz2018} A. Smoktunowicz and A. Smoktunowicz, \textsl{Set-theoretic solutions of the Yang–Baxter equation and new classes of R-matrices}, \href{https://doi.org/10.1016/j.laa.2018.02.001}{Linear Algebra Appl. \textbf{546}, 86--114 (2018)}, \texttt{arXiv:1704.03558}.

\bibitem{Smoktunowicz2020}A. Smoktunowicz, L. Vendramin, and R. Weston, \textsl{Combinatorial solutions to the reflection equation}, \href{https://doi.org/10.1016/j.jalgebra.2019.12.012}{J. Algebra \textbf{549} 268--290 (2020)}, \texttt{arXiv:1810.03341}.

\bibitem{Spitzer1970}F. Spitzer, \textsl{Interaction of Markov processes}, \href{https://doi.org/10.1016/0001-8708(70)90034-4}{Adv. Math. \textbf{5}, 246--290 (1970)}.

\bibitem{vanicat2017}M. Vanicat, \textsl{An integrabilist approach of out-of-equilibrium statistical physics models}, Phd thesis, Laboratoire d’Annecy-le-Vieux de Physique Théorique, 2017, \texttt{arXiv:1708.02440}.


\bibitem{Yang1967}C. Yang, \textsl{Some Exact Results for the Many-Body Problem in one Dimension with Repulsive Delta-Function Interaction}, \href{https://doi.org/10.1103/PhysRevLett.19.1312}{Phys. Rev. Lett. \textbf{19}, 1312--1315 (1967)}.

\end{thebibliography}
\end{document}